\documentclass[12pt,
 reprint,
superscriptaddress,
 amsmath,amssymb,
 aps,
 showkeys, 
]{revtex4-2}

\usepackage{graphicx}% Include figure files
\usepackage{dcolumn}% Align table columns on the decimal point
\usepackage{listings}
\usepackage{amsmath,amsthm, bbm}
\usepackage[export]{adjustbox}
\usepackage{bm}% bold math
\usepackage{subfigure}
\usepackage[dvipsnames]{xcolor}

\newcommand{\bra}[1]{\ensuremath{\left\langle#1\right|}}
\newcommand{\ket}[1]{\ensuremath{\left|#1\right\rangle}}

\newcommand{\braket}[2]{\ensuremath{\left\langle#1|#2\right\rangle}}
\newcommand{\ketbra}[2]{\ensuremath{\left|#1\rangle\langle#2\right|}}
\newtheorem{theorem}{Theorem}

\theoremstyle{definition}
\newtheorem{definition}{Definition}
\newtheorem{conjecture}{Conjecture}
\usepackage{tikz}
\usepackage{xcolor}
\newtheorem{proposition}{Proposition}
\usepackage{xpatch}

\makeatletter
\xpatchcmd{\proof}{\topsep6\p@\@plus6\p@\relax}{}{}{}
\makeatother

\definecolor{codegreen}{rgb}{0,0.6,0}
\definecolor{codegray}{rgb}{0.5,0.5,0.5}
\definecolor{codepurple}{rgb}{0.58,0,0.82}
\definecolor{backcolour}{rgb}{0.95,0.95,0.92}

\lstdefinestyle{mystyle}{
    backgroundcolor=\color{backcolour},   
    commentstyle=\color{codegreen},
    keywordstyle=\color{magenta},
    numberstyle=\tiny\color{codegray},
    stringstyle=\color{codepurple},
    basicstyle=\ttfamily\footnotesize,
    breakatwhitespace=false,         
    breaklines=true,                 
    captionpos=b,                    
    keepspaces=true,                 
    numbers=left,                    
    numbersep=5pt,                  
    showspaces=false,                
    showstringspaces=false,
    showtabs=false,                  
    tabsize=2
}

\begin{document}

\preprint{APS/123-QED}

\title{Quantum walks advantage on the dihedral group for uniform sampling problem\\}

\author{Shyam Dhamapurkar}
 \email{shyam18596@gmail.com }
\affiliation{Shenzhen Institute for Quantum Science and Engineering (SIQSE),
Southern University of Science and Technology, Shenzhen 518055, China}

\author{Yuhang Dang}
\affiliation{Shenzhen Institute for Quantum Science and Engineering (SIQSE),
Southern University of Science and Technology, Shenzhen 518055, China}
\author{Saniya Wagh}
\affiliation{Department of Mathematics, Tata Institute of Fundamental Research, India}

\author{Xiu-Hao Deng}
\email{ dengxh@sustech.edu.cn }
\affiliation{Shenzhen Institute for Quantum Science and Engineering (SIQSE),
Southern University of Science and Technology, Shenzhen 518055, China}
\affiliation{International Quantum Academy (SIQA), Futian District, Shenzhen 518048, China}

\date{\today}

\begin{abstract}
Random walk algorithms are crucial for sampling and approximation problems in statistical physics and theoretical computer science. The mixing property is necessary for Markov chains to approach stationary distributions and is facilitated by walks. Quantum walks show promise for faster mixing times than classical methods but lack universal proof, especially in finite group settings. Here, we investigate the continuous-time quantum walks on Cayley graphs of the dihedral group $D_{2n}$ for odd $n$, generated by the smallest inverse closed symmetric subset. We present a significant finding that, in contrast to the classical mixing time on these Cayley graphs, which typically takes at least order $\Omega(n^2 \log(1/2\epsilon))$, the continuous-time quantum walk mixing time on $D_{2n}$ is of order $O(n (\log n)^5 \log(1/\epsilon))$, achieving a quadratic improvement over the classical case. Our paper advances the general understanding of quantum walk mixing on Cayley graphs, highlighting the improved mixing time achieved by continuous-time quantum walks on $D_{2n}$. This work has potential applications in algorithms for a class of sampling problems based on non-abelian groups.    
 \end{abstract}


\maketitle

\section{{\em G\MakeLowercase{eneral introduction}}}

\textit{Background.- }Quantum computing promises algorithmic speedup compared to its classical counterpart \cite{Magniez05quantumalgorithms, Ambainis_element_distinctness, Childs_2004, Shor_Algo, grover1996fast}. Most of the quantum advantage campaigns are based on digital circuits ~\cite{simon1997power,lloyd2010quantum,farhi2014quantum}, such as Grover's amplitude amplification \cite{grover1996fast}, Shor's quantum Fourier transform \cite{Shor_Algo}, etc. However, analog or hybrid algorithms associated with sampling problems are attracting more and more interest in the NISQ era because of the efforts to push forward the application of quantum computing, such as QAOA, Metropolis-Hashing, and Hybrid Monte Carlo ~\cite{farhi2014quantum, Lemieux_2020, kim2012hybrid, cain2023quantum}. Among these algorithms, quantum walk emerges as a promising option for NISQ devices by showcasing exponential speedups of the hitting process on some graph structures for searching problems \cite{childs2003exponential}. The study of quantum walks has further extended to investigate mixing time on various graphs such as Erdos Renyi networks \cite{chakraborty2020fast}. This could help solve problems from a class of $\# P$ complete problems - the approximate sampling problems \cite{Richter_2007}. Notably, two specific problems, namely sampling a uniform permutation through card shuffling and reaching the Gibbs state using Glauber dynamics, are equivalent, e.g., for Refs.~\cite{ding2009mixing, saloff2004random}. It shows how probability theory and statistical mechanics are interlinked. The card shuffling problem can be viewed as a random walk on symmetric group $S_n$. Both discrete and continuous time quantum walks have been performed on the $S_n$ group to determine the mixing time~\cite{gerhardt2003continuous, banerjee2022discrete}.  

\begin{figure*}[t]
    \centering
    \subfigure[]{ \includegraphics[width=0.3\textwidth,keepaspectratio]{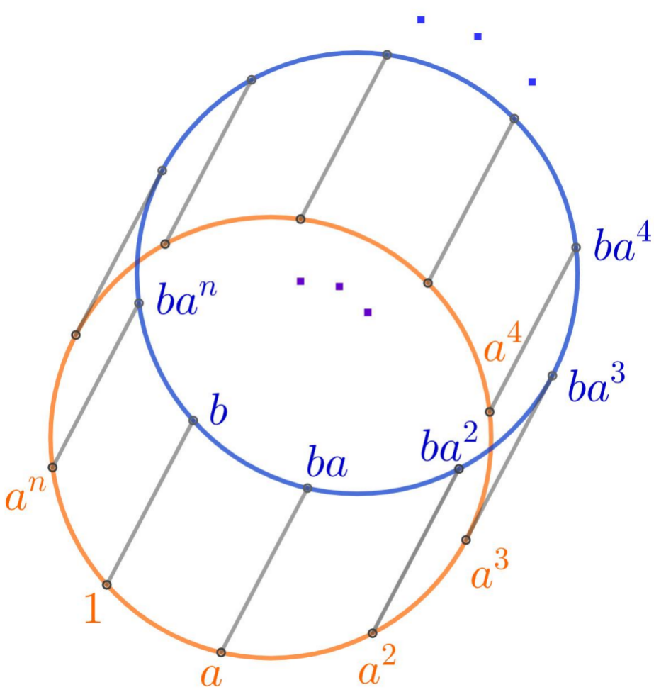}
    \label{Dg}
    }
    \subfigure[]{
     \includegraphics[width=0.43\textwidth,keepaspectratio]{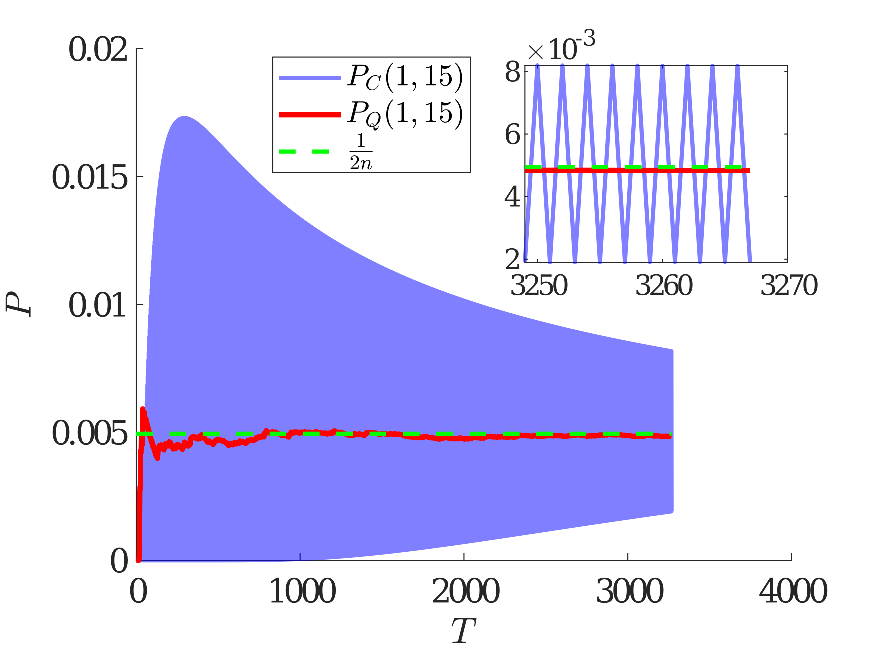}
     \label{P15}
    }
    \caption{ In Figure~\ref{Dg}, we depict the Cayley graphs we examine in this study, denoted as $\Gamma(D_{2n}, S = \{a, a^{-1}, b\})$. An edge $(g, h)$ exists between vertices $g$ and $h$ if $gh^{-1} \in S$, where $g$ and $h$ are elements of the group $D_{2n}$. In Figure~\ref{P15}, we provide an illustrative example for the case of $n = 101$. The quantum walk probability, denoted as \textcolor{red}{$P_Q(1,15)$}, approaches a uniform probability value of $1/2n = 0.005$ within a time of $T = O(n(\log{(n)})^5)$. On the other hand, the classical walk probability, represented as \textcolor{blue}{$P_C(1,15)$}, only fluctuates around $1/2n$. We establish the phenomenon in this work.}  
\end{figure*}
\textit{Issues.- }Extensive research has shown that quantum walks have exponential advantages compared to classical random walks on certain graphs ~\cite{venegas2012quantum,kadian2021quantum, Childs_2004, ambainis2001one, PR1, Richter_2007}. A few families of Cayley graphs have also been explored, such as Cayley graphs of Extraspecial groups~\cite{sin2022continuoustimequantumwalkscayley}. There is recent work on well-known graphs like the Johnson, Kneser, Grassman, and Rook graphs to sample (exact) uniformly using quantum walk \cite{wang2024unifyingquantumspatialsearch}. However, the quantum walk mixing time on Cayley graphs of the non-Abelian group is still not settled. Cayley graphs are an essential class of graphs for quantum walks because they are generated from groups and can be used to design quantum algorithms that exploit the symmetries and properties of the underlying group structure. Previously, properties like perfect state transfer, hitting time, and instantaneous uniform mixing have been verified on Cayley graphs of non-abelian groups~\cite{sin, Discrete_walk_Dn,cao2021perfect, three_state_QW}. Also, they can be used to study the quantum dynamics and transport phenomena on discrete structures, such as quantum coherence, entanglement, mixing, localization, and phase transitions~\cite{kendon2007decoherence}. Mathematically, a group is a set of elements equipped with an operation that satisfies closure, associativity, identity, and inverse properties~\cite{gallian2021contemporary}. Cayley graphs visually represent the symmetries of the group. The vertices of the Cayley graph are elements of the groups, and edges show how these elements relate to each other through the groups' operations.

\textit{Methodology.- }Proving the mixing time involves two main components: determining how long it takes for the mixing process to reach the limiting distribution (which may not be uniform) and exploring the possibility of uniform sampling from this distribution. Previously, it has been proved that continuous-time quantum walks with repeated measurements on certain Cayley graphs of a symmetric group $S_n$ do not converge to the uniform distribution~\cite{gerhardt2003continuous}. A remedy is given in reference~\cite{Richter_2007}, where Richter gave a double loop quantum walk algorithm for uniform sampling using quantum walks. He demonstrated that performing approximately $O(\log{(1/\epsilon)})$ iterations of the quantum walk $\text{e}^{-\text{i}Pt}$, where $P$ is a classical Markov chain on the underlying graph $\Gamma$, and selecting $t$ uniformly at random from the interval $[0, T]$, is sufficient to sample uniformly.  
This study focuses on the Cayley graphs of the dihedral group ($D_{2n}$). This group is symmetries of a regular polygon; reflection and rotation are the elements with composition operation. We use the same algorithms as Richter's to show the quadratic speedup on $D_{2n}$. By utilizing lower bounds on mixing time for regular graphs, we establish that classical random walks on Cayley graphs of $D_{2n}$ require at least $\Omega(n^2 \log(1/2\epsilon))$ time to achieve uniform mixing \cite{levin2017markov}. To estimate the mixing time of quantum walks, we employ the upper bound on mixing time provided in the reference~\cite{chakraborty2020fast}, which relies on eigenvalue gaps. We retrieve the adjacency matrix using the Ref.~\cite{GAO} method for the Cayley graph of dihedral groups to find the eigenvalues. The graph $\Gamma$ is generated by a symmetric inverse closed subset $S \in D_{2n}$ and is isomorphic to the semi-Cayley graph of an $n$-cycle, i.e., $\Gamma(D_{2n}, S) \cong \mathbb{Z}_n \rtimes \mathbb{Z}_2$. Our main result is that $O(n (\log{(n)})^5)$ time is sufficient to mix the continuous-time quantum walk with repeated measurements on the dihedral group towards a uniform distribution with $O(\log{(1/\epsilon)})$ iterations. To prove the main theorem, we propose a conjecture on the sum of the inverse of the difference in eigenvalue gaps for a subset of eigenvalues. We support the conjecture with simulations.   

 The work is organized as follows: The first section is dedicated to the preliminaries. Then, we discuss how to get the adjacency matrix for the Cayley graph $\Gamma(D_{2n}, S = \{a, a^{-1}, b\})$ from the semi-Cayley graphs. We give the general formulation to calculate the quantum walk amplitude. Subsequently, we calculate the limiting probability distribution on Cayley graphs of $D_{2n}$. Later, we analyze the mixing time on the Cayley graphs of dihedral groups and show that it is linear in the number of vertices on the graph. We conclude the sections with the results and future directions. The supplementary material includes the quantum walk algorithm, analysis of bounds from the main theorem, conjecture, and derivation of limiting distribution, respectively.

\section{{\em P\MakeLowercase{reliminaries}}}

This section introduces key definitions and propositions related to the Cayley graph, groups, Markov chains, and mixing time~\cite{levin2017markov} since the random walk is a special case of a Markov chain. The mixing time of a random or quantum walk refers to the duration it takes for the distribution of the walker to become $\epsilon$ distance close to its stationary distribution.

\begin{definition}{Cayley Graph:}
Consider a finite $G$ group and let $S \subseteq G$ be a symmetric subset of $G$, i.e., if $g \in S$, then $g^{-1} \in S$ for all $g \in G$. The Cayley graph is defined as $\Gamma(G, S)$, where elements of $G$ are the vertices of the graph $\Gamma$ and an edge $(g,h) \in \Gamma$ if and only if $gh^{-1} \in S$.   
\end{definition}

Suppose the size of $S$ is $d$, then for every vertex in $\Gamma$ has degree $d$. So, the Cayley graphs are $d$-regular graphs.

\begin{definition}{Conjugate:}
Consider a group $G$, let $g, h \in G$ be conjugate if there exists $r \in G$ such that $rgr^{-1} = h$, then $g$ is called a conjugate of $h$.   
\end{definition}

\begin{definition}{Semi-direct product:}
Consider a group $G$, $N$ as the normal subgroup, and $H$ as a proper subgroup of $G$. If $G = N H$ such that $N \cap H = \{e\}$, where e is the group's identity $G$, then $G$ is called a semi-direct product of $N$ and $H$. It can be written as $G = N \rtimes H$. 
\end{definition}

\begin{definition}{semi-Cayley graphs:}
Let $G$ be a group and $R$, $M$, $T$ be its subsets such that $R$ and $M$ are inverse closed and $e \notin R \cup M$. The semi- Cayley graph $\text{SC}(G; R,M,T)$ with the vertex set $G \times \{0,1\}$. To have an edge between vertices $(h,i)$ and $(g,j)$, one of the following holds:
\begin{itemize}
    \item $i = j = 0$  and $gh^{-1} \in R$;
    \item $i = j = 1$  and $gh^{-1} \in M$;
    \item $i = 0, j = 1$ and $gh^{-1} \in T$.
\end{itemize}

\end{definition}

A Markov chain is a stochastic process ${X_0, X_1, X_2, ...}$ with a countable set of states $S$, where the probability of transitioning from one state to another depends only on the current state. Mathematically, for any states $i, j \in S$ and any time steps $t \geq 0$, the Markov property can be expressed as:

$P(X_{t+1} = j | X_t = i, X_{t-1} = x_{t-1}, \dots, X_1 = x_1, X_0 = x_0 ) = P(X_{t+1} = j | X_t = i)$, where $P(X_{t+1} = j | X_t = i)$ represents the probability of transitioning from state $i$ to state $j$ in a one-time step.

\begin{definition}
    Markov chain $P$ has a stationary distribution $\pi$ implies that $P \pi = \pi$.
\end{definition}
 
\begin{definition}
  Consider an irreducible (strongly connected) and aperiodic (non-bipartite) Markov Chain $P$ with a stationary distribution $\pi$. The mixing time(also known as \emph{threshold mixing}) can be defined as follows:

\begin{equation}
   \tau_{\text{mix}} = \text{min} \Big \{ T: \frac{1}{2} \parallel P^{t} - \pi 1^\dagger \parallel_{1} \leq \frac{1}{2e} \hspace{1mm} \forall \hspace{1mm}t\geq T \Big \}, 
\end{equation}
where $\parallel . \parallel_{1}$ is a matrix 1-norm and $1^\dagger$ is all one row vector. 
\end{definition}

\begin{definition}
Given a Markov chain $P$,  
\[d(P) = \text{max}_{jj'}\frac{1}{2}\parallel P(.,j) - P(.,j') \parallel_{1},\]
is called the maximum pairwise column distance. The following inequality holds for $d(P)$.

 \begin{equation}
    \frac{1}{2}\parallel P - \pi 1^\dagger \parallel_{1} \leq d(P) \leq \parallel P - \pi 1^\dagger \parallel_{1}.
 \end{equation}    
\end{definition}

 The distance $d(P)$ is submultiplicative, i.e. 

 \begin{equation}
 \label{submu}
  d(P_{t+t'}) \leq d(P_t) d(P_{t'})   
 \end{equation}
  for any time $t$ and $t'$. This implies that $d(P_t) \leq d(P)^t$ and $d((P_t)^{t'}) \leq d(P_t)^{t'}$.

 \begin{proposition}\label{pro1}\cite{richter2007quantum}
     If $d(P_T) \leq 1/2e$, then $\parallel P_{T}^{T'} - \pi 1^\dagger \parallel_{1} \leq \epsilon$ for some time $T' = O(\log{(1/\epsilon)})$. 
 \end{proposition}
 
\section{{\em D\MakeLowercase{ihedral group}}}

In this section, we discuss the dihedral group $D_{2n}$, represented by symmetries of an $n$-regular polygon. We use the isomorphism given in the reference~\cite{GAO} between $\Gamma(D_{2n}, S)$ and semi-Cayley graph of $\mathbb{Z}_n$, allowing us to determine the adjacency matrix of $\Gamma$. The graph exhibits a unique structure, and the walks on $D_{2n}$ are equivalent to those on a specific semi-Cayley graph of $\mathbb{Z}_n$. The obtained adjacency matrix enables further analysis.

The dihedral group $D_{2n}$ is a finite group representing the symmetries of an $n$-regular polygon. It includes elements rotations and reflections and can be described abstractly as $\langle a, b | a^n = 1, b^2 = 1, bab = a^{-1} \rangle$. With $2n$ elements, $D_{2n}$ explicitly includes $\{1, a, b, ba, ba^2, \dots, ba^{n-1}, a^2, a^3, \dots, a^{n-1}\}$. To study quantum walks on $D_{2n}$, we construct the Cayley graph $\Gamma(D_{2n}, S)$ using the symmetric subset $S = \{a, a^{-1}, b \}$ as the generating set. The adjacency matrix $A(g, h)$ of $\Gamma(D_{2n}, S)$ is defined such that 
\begin{equation}\label{adj}
   A(g, h) =\begin{cases}
        1 & \text{ if } gh^{-1} \in S,\\
        0 & \text{ otherwise }.
    \end{cases}
\end{equation}
This representation provides the foundation for analyzing quantum walks on Cayley graphs associated with the dihedral group $D_{2n}$.

 To analyze continuous-time quantum walks on $\Gamma(D_{2n}, S)$, we can use an isomorphic counterpart, semi-Cayley graph on $\mathbb{Z}_n$ denoted as $\text{SC}(\mathbb{Z}_n; R, M, T)$, where $R = M = \{1, n-1\}$, $T = \{0\}$. This choice is advantageous because it simplifies the analysis. Here, $\mathbb{Z}_n$ represents a cyclic group of order $n$. Based on the findings in~\cite{GAO}, we can determine the spectral properties of the adjacency matrix of $\Gamma(D_{2n}, S)$. According to Lemma 4.2 in \cite{GAO} and the definition of semi-Cayley graphs, there exists an isomorphism $\phi$ between $\Gamma(D_{2n}, S)$ and $\text{SC}(\mathbb{Z}_n; R, M, T)$. The isomorphism is defined as $\phi(a^r) = (r,0)$ and $\phi(ba^r) = (n-r,1)$ for $r \in [0,n-1]$. Consequently, performing continuous-time quantum walks on the graph $\text{SC}(\mathbb{Z}_n; R, M, T)$ is equivalent to conducting the same walks on $\Gamma(D_{2n}, S)$. The adjacency matrix is of $\text{SC}(\mathbb{Z}_n; R, M, T)$ is given as follows:

\begin{equation}\label{Eq;Adj}
    A = \begin{bmatrix}
    W_{n} + W_{n}^{n-1} & I \\
    I & W_{n} + W_{n}^{n-1}
    \end{bmatrix}.
\end{equation}

Here, $W_{n}$ is $n \times n$ a circulant matrix given below.

 \begin{equation*}
     W_{n} = \begin{bmatrix}
          0 & 1 & 0 & \dots & 0 \\
          0 & 0 & 1 & \dots & 0 \\
        \vdots & \vdots & \vdots & \vdots & \vdots \\  
         1 & 0 & 0 & \dots & 0
          \end{bmatrix},
 \end{equation*}
 
 with eigenvalues $w^j$ and eigenvectors 
 $v_j^{\top} = (1/\sqrt{n})[1, w^j, w^{2j}, \dots, w^{(n-1)j}]^\top$ for $0 \leq j \leq n-1$, where $w = \text{e}^{(2 \pi \text{i}/n)}$ and $\top$ is transpose. In the next section, we will discuss how to define a unitary evolution operator to do a quantum walk on $\Gamma(D_{2n}, S)$.
 
\section{{\em Q\MakeLowercase{uantum walk on $D_{2n}$}}}

In this section, we state the eigenspectrum of the adjacency matrix $A$ given in Eq.~\eqref{Eq;Adj}. Next, we provide the lower bound for the classical mixing time of a random walk on $\Gamma(D_{2n}, S)$. Afterwards, we discuss the time-averaged quantum walk probability and the limiting distribution of quantum walks on $\Gamma$.

The simple random walk matrix for regular graphs is simply the normalized adjacency matrix of the graph, i.e., $\Bar{A} = \frac{A}{3}$. We use normalized adjacency matrix $\Bar{A}$, which gives us the following normalized $2n$ eigenvalues and eigenvectors, respectively. 

\begin{equation}\label{lambdaaj}
  \lambda_j = \left(1 + 2 \cos{(2\pi j/n )}\right)/3,   \end{equation}
  
  \begin{equation}\label{alphaj}
      \ket{x_j} =\frac{1}{\sqrt{2n}} \begin{bmatrix}
           v_j & v_j
      \end{bmatrix}^\dagger,
  \end{equation}
 for $0 \leq j \leq (n-1)$, and 

\begin{equation}\label{lambdabj}
  \lambda_{j} = \left(2 \cos{\left(2 \pi (j-n)/n\right)} - 1\right)/3,   
\end{equation}
  
  \begin{equation}\label{betaj}
      \ket{y_{j}} =\frac{1}{\sqrt{2n}} \begin{bmatrix}
           v_j  & -v_j
      \end{bmatrix}^\dagger,
  \end{equation}
for $n \leq j \leq (2n-1)$.

To prove the classical mixing time lower bound on $\Gamma(D_{2n}, S)$, we use the Theorem 12.5 given in Ref.~\cite{levin2017markov} which states that for transition matrix $P$ of a reversible, irreducible Markov chain then $\tau_{\text{mix}}(\epsilon) \geq (1/(1- \lambda_2)-1) \log{\left(1/2\epsilon\right)},$ where $\lambda_2$ is the second largest eigenvalue. 
    
Upon examining Eq.\eqref{Eq;Adj}, it becomes evident that the simple random walk $\Bar{A}$ on $\Gamma(D_{2n}, S)$ possesses key properties, namely symmetry (reversibility) and strong connectivity (irreducibility). Additionally, the Cayley graph $\Gamma$ is regular, resulting in a uniform stationary distribution $\pi$, Ref.~\cite{levin2017markov}. We can determine the second largest eigenvalue of $\Bar{A}$ as $\lambda_2 = \left(1 + 2 \cos{(2\pi/n )}\right)/3$ using Eq.\eqref{lambdaaj}. By employing the inequality $1-\cos(x) \leq x^2/2,$ we find that \[\tau_{\text{mix}}(\epsilon) \geq (3n^2/4\pi^2-1) \log{\left(1/2\epsilon\right)}.\] Consequently, the classical mixing time on $\Gamma(D_{2n}, S)$ is at least $\Omega\left(n^2 \log{\left(1/2\epsilon\right)}\right)$.

   Now, we define the unitary operator. Based on the eigen-spectrum of $\Bar{A}$, the continuous-time quantum walk operator $U(t) = \text{e}^{\text{i}\Bar{A}t}$ can be written as
  
  \begin{equation}
     U(t) = \sum_{j = 0}^{n-1} \text{e}^{\text{i} \lambda_j t} \ketbra{x_j}{x_j} + \sum_{j =n}^{2n-1} \text{e}^{\text{i} \lambda_{j} t} \ketbra{y_{j}}{y_{j}}.
  \end{equation}
  
  Let us define $X_j :=\ketbra{x_j}{x_j}$ and $Y_j :=\ketbra{y_{j}}{y_{j}}$ for $0 \leq j \leq (2n-1)$ respectively. Then the probability to go from some vertex $\ket{p}$ to another vertex $\ket{q}$ on the graph $\Gamma(D_{2n}, S)$  in time $t$ is given by
\begin{equation}
\begin{aligned}\label{Eq:prob}
    P_{t}(p,q) &=\Bigg|\frac{1}{2n}\bra{q}\sum_{j = 0}^{n-1} \text{e}^{\text{i}t(2\cos{(2 \pi j/n)}+1)/3} X_j \\
    &+ \sum_{j = n}^{2n-1} \text{e}^{\text{i}t(2\cos{(2 \pi (j-n)/n)}-1)/3} Y_j\ket{p}\Bigg|^2.
\end{aligned}    
\end{equation}

For each $1 \leq p,q \leq 2n$, we get $P_{t}(p,q)$, and that gives us $P_{t}$ matrix, a quantum-generated stochastic matrix. Since the evolution $U(t)$ is unitary, we know that for large times, it will not converge to any specific distribution. Hence, we do probability averaging over an interval of time. It results in the time-averaged probability matrix $\Bar{P}_{T}$, where each $(p,q)$ entry is given by 

\begin{equation}\label{time-avg}
    \Bar{P}_T(p,q) = \frac{1}{T} \int_{0}^{T} P_{t}(p,q) \text{d}t.
\end{equation}

The long-term behavior of this quantum walks $\Bar{P}_T$ always fluctuates around its limiting distribution, which is stationary. We denote the limiting distribution by $\Pi$. Calculating the entries $\Pi(p,q)$ of limiting distribution of a quantum walk on $\Gamma$ is straightforward (refer to supplementary material D). When we take the limit of $T \rightarrow \infty$ in Eq.~\eqref{time-avg}, we find that $P_{T \rightarrow \infty}(p,q)$ is equal to

\begin{equation}\label{Pi}
\Pi(p,q) =\begin{cases} 
\frac{1}{2n} + \frac{n-1}{2n^2} & p =q \text{ or } q-p = n, \\
\frac{1}{2n} - \frac{1}{2n^2} & p \neq q \text{ and } q-p\neq n.
\end{cases}
\end{equation} 
 
It is worth noting that the distribution $\Pi$ described in Eq.~\eqref{Pi} is non-uniform. To sample uniformly from $P_t$(subsequently $\Bar{P}_T$ and $\Pi$), has to have uniform distribution. We show that the following Theorem~\ref{uniformd} holds for the Markov chain $P = \Bar{A}$ and $\Pi$.  
\begin{theorem}\cite{Richter_2007}~\label{uniformd}
  If $P$ is a symmetric irreducible Markov chain on $N$ states, then each entry of $\Pi$ bounded below by $1/N^2$; in particular, $\Pi$ is ergodic. Moreover, each $P_t$ is symmetric and, hence, has a uniform stationary distribution.  
\end{theorem}

By inspection, it is clear that each entry of $\Pi$ given in Eq.~\eqref{Pi} is bounded below by $1/(2n)^2$ and $P_t$ in Eq.~\eqref{Eq:prob} is symmetric since $P_t(p,q) = P_t(q,p)$ (and so is $\Bar{P}_T$). Hence, $P_t$ has a uniform stationary distribution. We utilize the double-loop quantum algorithm to achieve uniform sampling, as outlined in Supplementary Material A. This algorithm was originally introduced by Richter in their work \cite{Richter_2007}, and it exhibits a logarithmic dependence on $1/\epsilon$, where $\epsilon$ represents the desired accuracy or precision. This algorithm essentially runs a classical random walk for a duration of $T' = O(\log{(1/\epsilon)})$ using the quantum-generated stochastic matrix $\Bar{P}_T$, (given in Eq.~\eqref{time-avg}).
 We are interested in the minimum time $\tau_{\text{mix}} = T$ such that 
\begin{equation}\label{boundlim}
   \parallel \Bar{P}_T - \Pi \parallel_1 \leq \frac{1}{2e}.  
\end{equation}
 Then, using proposition~\ref{pro1}, $T' = O(\log{(1/\epsilon)})$, repetitions of this walk are adequate for achieving uniform sampling. The following section discusses the quantum mixing time bound based on the inverse sum of eigenvalue gaps to achieve Eq.~\eqref{boundlim}. To do that, we state a conjecture for the subset of eigenvalues of $\Bar{A}$.  

\section{{\em M\MakeLowercase{ixing time bound and conjecture}}}

This section focuses on the quantum mixing time bound, utilizing the eigenvalue gaps of $\Bar{A}$. We present the general quantum mixing time bound based on these eigenvalues. Subsequently, we provide specific bounds for our case and propose a conjecture for certain eigenvalue gaps. Finally, we establish the main result.

Given $\Bar{P}_T$, the quantum mixing bound based on eigenvalues of $\Bar{A}$ (Ref.~\cite{chakraborty2020fast}) on the L.H.S. of Eq.~\eqref{boundlim} is given as follows:   

\begin{equation}\label{qbound}
    \parallel \Bar{P}_T - \Pi \parallel_1 \leq \sum_{\lambda_j \neq \lambda_k} \frac{2 |\braket{x_j \text{ or } y_j}{p}|.  |\braket{p}{x_k \text{ or } y_k}|}{T |\lambda_j - \lambda_k|},
\end{equation}

where without loss of generality $\ket{p}$ is initial state, and $\{ \lambda_j, \ket{x_j},\ket{y_j}\}$ is the eigen spectrum of $\Bar{A}$. From Eq.~\eqref{alphaj} and Eq.~\eqref{betaj}, for $1 \leq p \leq 2n$ we have

\begin{equation}
    |\braket{x_j \text{ or } y_j}{p}|.|\braket{p}{x_k \text{ or } y_k}| = \frac{1}{2n}.
\end{equation}

Using Eq.~\eqref{qbound} and the above calculations, the quantity we need to bound is the following:

\begin{equation}
  \frac{1}{nT} \sum_{\lambda_j \neq \lambda_k} \frac{1}{|\lambda_j - \lambda_k|}.
\end{equation}

We initially partitioned the set of $2n$ eigenvalues into two subsets based on their indices: the first subset consists of eigenvalues given by Eq.~\eqref{alphaj} for $0 \leq j \leq n-1$, and the second subset comprises eigenvalues given by Eq.~\eqref{betaj} for $n \leq j \leq 2n-1$.

To further identify distinct eigenvalues within these subsets, we introduce index subsets $C_1$ and $C_2$ as follows: For $j \in C_1 = [0, \frac{n-1}{2}]$, the eigenvalues satisfy $1 \geq \lambda_j > -\frac{1}{3}$. For $j \in C_2 = [n, \frac{3n-1}{2}]$, the eigenvalues satisfy $\frac{1}{3} \geq \lambda_j > -1$.

Additionally, we define index subsets $C_{1'} = [\frac{n+1}{2}, n-1]$ and $C_{2'} = [\frac{3n+1}{2}, 2n-1]$. In $C_{1'}$, the eigenvalues fall within the range $(1, -\frac{1}{3})$, while in $C_{2'}$, the eigenvalues fall within the range $(\frac{1}{3}, -1)$. Notably, $C_{1'}$ and $C_{2'}$ each contain $n-2$ repeated eigenvalues from the $C_1$ and $C_2$ subsets, respectively, resulting in distinct sets of eigenvalues.
\begin{equation}\label{Eq: eigsum}
 \begin{aligned}
    \sum_{\lambda_j \neq \lambda_k} \frac{1}{|\lambda_j - \lambda_k|} & = 2\sum_{j \in C_1} \Big( \sum_{k \in C_2} \frac{1}{|\lambda_j - \lambda_k|} + \sum_{k \in C_{1'}} \frac{1}{|\lambda_j - \lambda_k|}\\
    &+\sum_{k \in C_{2'}} \frac{1}{|\lambda_j - \lambda_k|}+ \sum_{k \in C_1, \lambda_j \neq \lambda_k} \frac{1}{|\lambda_j - \lambda_k|} \Big)\\
    &+ 2\sum_{j \in C_2} \Big( \sum_{k \in C_1} \frac{1}{|\lambda_j - \lambda_k|} + \sum_{k \in C_{1'}} \frac{1}{|\lambda_j - \lambda_k|} \\
    &+\sum_{k \in C_{2'}} \frac{1}{|\lambda_j - \lambda_k|}+ \sum_{k \in C_2, \lambda_j \neq \lambda_k} \frac{1}{|\lambda_j - \lambda_k|} \Big).
\end{aligned}   
\end{equation}

We simplify Eq.~\eqref{Eq: eigsum} by redefining the ranges of indices $k$ and $j$ to only involve the index sets $C_1$ and $C_2$ (by doing the change of variable $k \rightarrow k+n$ and $j \rightarrow n-j$ ). We get the following equation.

\begin{equation}\label{eqsum}
\begin{aligned}
    \sum_{\lambda_j \neq \lambda_k} \frac{1}{|\lambda_j - \lambda_k|} & = 8 \sum_{j \in C_1} \sum_{k \in C_2} \frac{1}{|\lambda_j - \lambda_k|}\\
    &+ 
    4\sum_{j \in C_1} \sum_{k \in C_1,\lambda_j \neq \lambda_k} \frac{1}{|\lambda_j - \lambda_k|}\\
    &+ 4\sum_{j \in C_2} \sum_{k \in C_2,\lambda_j \neq \lambda_k} \frac{1}{|\lambda_j - \lambda_k|}. 
\end{aligned}    
\end{equation}
Now, we calculate the bound on each sum from Eq.~\eqref{eqsum}. Let us tackle the first sum by simplifying it as follows:
\begin{equation}\label{mainsum}
\begin{aligned}
   &\sum_{j \in C_1} \sum_{k \in C_2} \frac{1}{|\lambda_j - \lambda_k|}\\
   &= 3 \sum_{j \in C_1} \sum_{k \in C_2} \frac{1}{|2\cos{(2\pi j/n)} - 2\cos{(2\pi (k-n)/n)}+2 |} \\
   & = 3\sum_{j \in C_1} \sum_{k \in C_1} \frac{1}{|2\cos{(2\pi j/n)} - 2\cos{(2\pi k)/n)}+2 |}\\
   & = \frac{3}{2}\sum_{j \in C_1} \sum_{k \in C_1} \frac{1}{|\cos{(2\pi j/n)} - \cos{(2\pi k)/n)}+1 |}.
\end{aligned}   
\end{equation}

We divide Eq.~\eqref{mainsum} into four sums based on the range of $j$ and $k$ as follows: 

\begin{equation}
  \begin{aligned}
 &\frac{3}{2}\sum_{j \in C_1} \sum_{k \in C_1} \frac{1}{|\cos{(2\pi j/n)} - \cos{(2\pi k)/n)}+1 |} \\
 &= Su_{1} + Su_2 +Su_3 +Su_4,   
\end{aligned}  
\end{equation}

where 
\begin{equation}
 \begin{aligned}
 & Su_1 = \frac{3}{2}\sum_{j = 0}^{\lfloor\frac{n}{4}\rfloor} \sum_{k = 0}^{\lfloor\frac{n}{4}\rfloor} \frac{1}{|\cos{(2\pi j/n)} - \cos{(2\pi k)/n)}+1 |},\\
&Su_2 = \frac{3}{2}\sum_{j = 0}^{\lfloor\frac{n}{4}\rfloor} \sum_{k = \lceil\frac{n}{4}\rceil}^{(n-1)/2} \frac{1}{|\cos{(2\pi j/n)} - \cos{(2\pi k)/n)}+1 |},\\
&Su_3 = \frac{3}{2}\sum_{j = \lceil\frac{n}{4}\rceil}^{(n-1)/2} \sum_{k = 0}^{\lfloor\frac{n}{4}\rfloor} \frac{1}{|\cos{(2\pi j/n)} - \cos{(2\pi k)/n)}+1 |},\\
& Su_4 = \frac{3}{2}\sum_{j = \lceil\frac{n}{4}\rceil}^{(n-1)/2} \sum_{k = \lceil\frac{n}{4}\rceil}^{(n-1)/2} \frac{1}{|\cos{(2\pi j/n)} - \cos{(2\pi k)/n)}+1 |}.
\end{aligned}   
\end{equation}

We bound each $Su_{l}$ for $1\leq l \leq 4$ separately. The proofs of bounds on $Su_{1}$, $Su_{2}$, and $Su_{4}$ are given in Supplementary material B. The bounds are as follows:   

\begin{equation}
\begin{aligned}\label{sumbound}
  Su_{1} &\leq \frac{3}{8} n^2 \log{(n)},\\
  Su_{2} &\leq \frac{3}{32} n^2,\\
    Su_{4} &\leq \frac{3}{32} n^2 \log{(n)}.
\end{aligned}    
\end{equation}

We could not prove the upper bound on $Su_{3}$ rigorously. Hence, we propose conjecture~\ref{conj}. We give a numerical argument for the bound on $Su_{3}$. The following conjecture is for $n = 4p+1$; similarly, we do for $n = 4p+3$ (refer to Supplementary material C).  

\begin{conjecture}\label{conj}
   Consider $n = 4p+1$ type, where $p \in \mathbb{Z}_{+}$, let $\alpha = \arccos{\Big(1-\sin{\Big(\frac{2\pi}{n}(b+\frac{3}{4})\Big)} \Big)}$, and $N(\alpha) = \lfloor \frac{n}{2\pi} \alpha \rfloor$ for $b \in[0, p-1]$ then

\begin{equation}
\begin{aligned}
Su_{3} &\leq f(n) = \sum_{b=0}^{p-1} f_{\alpha}(b) \\
&\leq 100 n^2 (\log{(n)})^5,  
   \end{aligned}    
\end{equation}
   
  where  \begin{align*}
    f_{\alpha}(b) &= \frac{\pi}{2\alpha}\Bigg[\frac{1}{\alpha} + \frac{1}{\alpha-\frac{2\pi}{n}N(\alpha)} + \frac{1}{\frac{2\pi}{n}(N(\alpha)+1)-\alpha}\Bigg]\\
    &+ \frac{n}{4\alpha} \ln{\Bigg[\frac{\frac{\pi^2}{2}}{ (\alpha-\frac{2\pi}{n}N(\alpha)) (\frac{2\pi}{n}(N(\alpha)+1)-\alpha)} \Bigg]}.  
\end{align*}
\end{conjecture}

By conducting numerical simulations, we provide evidence supporting the validity of the conjecture (see Supplementary material C). Consequently, we establish a bound on $Su_{3}$ as $100 n^2(\log{(n)})^5$.

Lastly, we analyze the other two sums from Eq.~\eqref{eqsum} where $\lambda_j \neq \lambda_k$. We show that (refer to Supplementary Material B, case 5 for the proof.)

\begin{equation}\label{S5}
     \sum_{j \in C_1} \sum_{k \in C_1,\lambda_j \neq \lambda_k} \frac{1}{|\lambda_j - \lambda_k|} \leq \Big(\frac{8n}{\pi} \log{(n)}\Big)^2 .
\end{equation}

Similarly, we prove

\begin{equation}
    \sum_{j \in C_2} \sum_{k \in C_2,\lambda_j \neq \lambda_k} \frac{1}{|\lambda_j - \lambda_k|} \leq  \Big(\frac{8n}{\pi} \log{(n)}\Big)^2. 
\end{equation}

Now, we state our main theorem and give the proof using the bounds and conjecture mentioned above. 

\begin{theorem}\label{maintheo}
 For a time $T$ of order $O(n(\log{(n)})^5)$ and $T' = O\left(\log{(1/\epsilon)}\right)$ iterations, the repeated continuous-time quantum walk on the graph $\Gamma(D_{2n}, S)$ with $S = \{a, a^{-1}, b\}$ converges to the uniform distribution when $n$ is odd.
\end{theorem}

\begin{proof}
We combine the bounds given in Eq.~\eqref{sumbound},~\eqref{S5}, and Conjecture~\ref{conj} to give the mixing time bound as follows.   

\begin{align*}
     \parallel \Bar{P}_T - \Pi \parallel_1 &\leq 
     \frac{1}{nT} \Big(3n^2 \log{(n)} + \frac{3}{4} n^2+\frac{3}{4} n^2\\
    & + 800 n^2 (\log{(n)})^5+\frac{256}{\pi} (n \log{(n)})^2 \\
    &+ \frac{256}{\pi} (n \log{(n)})^2 \Big)\\
     &\leq \frac{1}{T} \Big(3 n\log{(n)} +2n +800 n (\log{(n)})^5  \\
     &+163 n (\log{(n)})^2  \Big).
\end{align*}

For $T =  4800 n(\log{(n)})^5$, 

\begin{align*}
 \parallel \Bar{P}_T - \Pi \parallel_ 1 &\leq \Big(\frac{1}{1600 (\log{(n)})^4} +\frac{1}{2400 (\log{(n)})^5} + \frac{1}{6}\\
 &+\frac{163}{4800(\log{(n)})^3}   \Big)\\
 & \leq \Big(\frac{1}{6} + \frac{1}{10(\log{(n)})^3} \Big).
\end{align*}
For all $n \geq 100$, $\parallel \Bar{P}_T - \Pi \parallel_1 \leq 1/2e$. According to Theorem~\ref{uniformd}, the matrix $\Bar{P}_T$ admits a uniform stationary distribution $\pi$. It is apparent that $\pi$ acts as an eigenvector of $\Bar{P}_T$, corresponding to an eigenvalue of 1. This leads to the representation of $\Bar{P}_T$ as $\Bar{P}_T = \ketbra{\pi}{\pi} + \sum_{j= 2}^{2n^2} v_j \ketbra{v_j}{v_j}$, where $v_j$ are remaining eigenvalues, each associated with an eigenvector $\ket{v_j}$ and all being less than one. Consequently, achieving uniform sampling necessitates a sufficient number of repetitions $T'$ of $\Bar{P}_T$. To attain an $\epsilon$-closeness to $\pi$, it is required that $(1/2e)^{T'} \leq \epsilon$, which translates to $T' \geq \log{(1/\epsilon)}$.   
\end{proof}

\section{{\em S\MakeLowercase{ummary and outlook}}}
In this study, we focus on the quantum mixing time of Cayley graphs associated with $D_{2n}$. We present an upper bound for the mixing time of a continuous-time quantum walk with repeated measurements on $D_{2n}$. Our results show that within $O(n (\log{(n)})^5)$ time, the quantum walk approaches the limiting distribution. By performing $O(\log{(1/\epsilon)})$ iterations, we achieve uniform sampling, surpassing the classical lower bound of $\Omega(n^2 \log{(1/2\epsilon)})$.

Additionally, we put forward a conjecture that relates to the sum of a subset of the reciprocals of eigenvalue gaps. This conjecture is supported by numerical evidence. Moreover, we highlight the quadratic advantage offered for the classical shuffling problem with the $D_{2n}$ group. This study raises the question of the potential advantages of quantum walks on finite groups in general.

Overall, our work expands the range of classical Markov chain Monte Carlo processes in which quantum walks with repeated measurements exhibit a speedup advantage. This study encourages further investigation into potential applications, especially sampling algorithms. Also, testing graph isomorphism is a hard problem in general, with applications in chemistry, network analysis, and computer vision. Quantum walks on Cayley graphs can construct canonical forms of graphs, which are unique representations that can be compared efficiently with a speed faster than that of a random walk. 

 \noindent {\bf {\em Acknowledgements.---}} 
We want to express our gratitude to Pranab Sen (Tata Institute of Fundamental Research, Mumbai) and Upendra Kapshikar (CQT) for giving us valuable insights on the initial part of the work. We also thank Ganesh Kadu and Hemant Bhate (SPPU) for their valuable discussions and suggestions on the finite groups. We would also like to thank Xiaolong \textcolor{blue}{Zhu} for the discussions.

This work was supported by the Key-Area Research and Development Program of Guang-Dong Province (Grant No. 2018B030326001), the National Natural Science Foundation of China (U1801661), Shenzhen Science and Technology Program (KQTD20200820113010023).
 
\nocite{*}
\bibliographystyle{apsrev4-1}
\bibliography{apssamp}

\begin{thebibliography}{35}%
\makeatletter
\providecommand \@ifxundefined [1]{%
 \@ifx{#1\undefined}
}%
\providecommand \@ifnum [1]{%
 \ifnum #1\expandafter \@firstoftwo
 \else \expandafter \@secondoftwo
 \fi
}%
\providecommand \@ifx [1]{%
 \ifx #1\expandafter \@firstoftwo
 \else \expandafter \@secondoftwo
 \fi
}%
\providecommand \natexlab [1]{#1}%
\providecommand \enquote  [1]{``#1''}%
\providecommand \bibnamefont  [1]{#1}%
\providecommand \bibfnamefont [1]{#1}%
\providecommand \citenamefont [1]{#1}%
\providecommand \href@noop [0]{\@secondoftwo}%
\providecommand \href [0]{\begingroup \@sanitize@url \@href}%
\providecommand \@href[1]{\@@startlink{#1}\@@href}%
\providecommand \@@href[1]{\endgroup#1\@@endlink}%
\providecommand \@sanitize@url [0]{\catcode `\\12\catcode `\$12\catcode `\&12\catcode `\#12\catcode `\^12\catcode `\_12\catcode `\%12\relax}%
\providecommand \@@startlink[1]{}%
\providecommand \@@endlink[0]{}%
\providecommand \url  [0]{\begingroup\@sanitize@url \@url }%
\providecommand \@url [1]{\endgroup\@href {#1}{\urlprefix }}%
\providecommand \urlprefix  [0]{URL }%
\providecommand \Eprint [0]{\href }%
\providecommand \doibase [0]{http://dx.doi.org/}%
\providecommand \selectlanguage [0]{\@gobble}%
\providecommand \bibinfo  [0]{\@secondoftwo}%
\providecommand \bibfield  [0]{\@secondoftwo}%
\providecommand \translation [1]{[#1]}%
\providecommand \BibitemOpen [0]{}%
\providecommand \bibitemStop [0]{}%
\providecommand \bibitemNoStop [0]{.\EOS\space}%
\providecommand \EOS [0]{\spacefactor3000\relax}%
\providecommand \BibitemShut  [1]{\csname bibitem#1\endcsname}%
\let\auto@bib@innerbib\@empty
%</preamble>
\bibitem [{\citenamefont {Magniez}\ \emph {et~al.}(2005)\citenamefont {Magniez}, \citenamefont {Santha},\ and\ \citenamefont {Szegedy}}]{Magniez05quantumalgorithms}%
  \BibitemOpen
  \bibfield  {author} {\bibinfo {author} {\bibfnamefont {F.}~\bibnamefont {Magniez}}, \bibinfo {author} {\bibfnamefont {M.}~\bibnamefont {Santha}}, \ and\ \bibinfo {author} {\bibfnamefont {M.}~\bibnamefont {Szegedy}},\ }in\ \href@noop {} {\emph {\bibinfo {booktitle} {PROCEEDINGS OF SODA’05}}}\ (\bibinfo {year} {2005})\ pp.\ \bibinfo {pages} {1109--1117}\BibitemShut {NoStop}%
\bibitem [{\citenamefont {Ambainis}(2007)}]{Ambainis_element_distinctness}%
  \BibitemOpen
  \bibfield  {author} {\bibinfo {author} {\bibfnamefont {A.}~\bibnamefont {Ambainis}},\ }\href@noop {} {\bibfield  {journal} {\bibinfo  {journal} {SIAM Journal on Computing}\ }\textbf {\bibinfo {volume} {37}},\ \bibinfo {pages} {210} (\bibinfo {year} {2007})}\BibitemShut {NoStop}%
\bibitem [{\citenamefont {Childs}\ and\ \citenamefont {Goldstone}(2004)}]{Childs_2004}%
  \BibitemOpen
  \bibfield  {author} {\bibinfo {author} {\bibfnamefont {A.~M.}\ \bibnamefont {Childs}}\ and\ \bibinfo {author} {\bibfnamefont {J.}~\bibnamefont {Goldstone}},\ }\href@noop {} {\bibfield  {journal} {\bibinfo  {journal} {Physical Review A}\ }\textbf {\bibinfo {volume} {70}},\ \bibinfo {pages} {022314} (\bibinfo {year} {2004})}\BibitemShut {NoStop}%
\bibitem [{\citenamefont {Shor}(1997)}]{Shor_Algo}%
  \BibitemOpen
  \bibfield  {author} {\bibinfo {author} {\bibfnamefont {P.~W.}\ \bibnamefont {Shor}},\ }\href {\doibase 10.1137/S0097539795293172} {\bibfield  {journal} {\bibinfo  {journal} {SIAM J. Comput.}\ }\textbf {\bibinfo {volume} {26}},\ \bibinfo {pages} {1484–1509} (\bibinfo {year} {1997})}\BibitemShut {NoStop}%
\bibitem [{\citenamefont {Grover}(1996)}]{grover1996fast}%
  \BibitemOpen
  \bibfield  {author} {\bibinfo {author} {\bibfnamefont {L.~K.}\ \bibnamefont {Grover}},\ }in\ \href@noop {} {\emph {\bibinfo {booktitle} {Proceedings of the twenty-eighth annual ACM symposium on Theory of computing}}}\ (\bibinfo {year} {1996})\ pp.\ \bibinfo {pages} {212--219}\BibitemShut {NoStop}%
\bibitem [{\citenamefont {Simon}(1997)}]{simon1997power}%
  \BibitemOpen
  \bibfield  {author} {\bibinfo {author} {\bibfnamefont {D.~R.}\ \bibnamefont {Simon}},\ }\href@noop {} {\bibfield  {journal} {\bibinfo  {journal} {SIAM journal on computing}\ }\textbf {\bibinfo {volume} {26}},\ \bibinfo {pages} {1474} (\bibinfo {year} {1997})}\BibitemShut {NoStop}%
\bibitem [{\citenamefont {Lloyd}(2010)}]{lloyd2010quantum}%
  \BibitemOpen
  \bibfield  {author} {\bibinfo {author} {\bibfnamefont {S.}~\bibnamefont {Lloyd}},\ }in\ \href@noop {} {\emph {\bibinfo {booktitle} {APS March Meeting Abstracts}}},\ Vol.\ \bibinfo {volume} {2010}\ (\bibinfo {year} {2010})\ pp.\ \bibinfo {pages} {D4--002}\BibitemShut {NoStop}%
\bibitem [{\citenamefont {Farhi}\ \emph {et~al.}(2014)\citenamefont {Farhi}, \citenamefont {Goldstone},\ and\ \citenamefont {Gutmann}}]{farhi2014quantum}%
  \BibitemOpen
  \bibfield  {author} {\bibinfo {author} {\bibfnamefont {E.}~\bibnamefont {Farhi}}, \bibinfo {author} {\bibfnamefont {J.}~\bibnamefont {Goldstone}}, \ and\ \bibinfo {author} {\bibfnamefont {S.}~\bibnamefont {Gutmann}},\ }\href@noop {} {\bibfield  {journal} {\bibinfo  {journal} {arXiv preprint}\ } (\bibinfo {year} {2014})}\BibitemShut {NoStop}%
\bibitem [{\citenamefont {Lemieux}\ \emph {et~al.}(2020)\citenamefont {Lemieux}, \citenamefont {Heim}, \citenamefont {Poulin}, \citenamefont {Svore},\ and\ \citenamefont {Troyer}}]{Lemieux_2020}%
  \BibitemOpen
  \bibfield  {author} {\bibinfo {author} {\bibfnamefont {J.}~\bibnamefont {Lemieux}}, \bibinfo {author} {\bibfnamefont {B.}~\bibnamefont {Heim}}, \bibinfo {author} {\bibfnamefont {D.}~\bibnamefont {Poulin}}, \bibinfo {author} {\bibfnamefont {K.}~\bibnamefont {Svore}}, \ and\ \bibinfo {author} {\bibfnamefont {M.}~\bibnamefont {Troyer}},\ }\href {\doibase 10.22331/q-2020-06-29-287} {\bibfield  {journal} {\bibinfo  {journal} {Quantum}\ }\textbf {\bibinfo {volume} {4}},\ \bibinfo {pages} {287} (\bibinfo {year} {2020})}\BibitemShut {NoStop}%
\bibitem [{\citenamefont {Kim}\ \emph {et~al.}(2012)\citenamefont {Kim}, \citenamefont {Esler}, \citenamefont {McMinis}, \citenamefont {Morales}, \citenamefont {Clark}, \citenamefont {Shulenburger},\ and\ \citenamefont {Ceperley}}]{kim2012hybrid}%
  \BibitemOpen
  \bibfield  {author} {\bibinfo {author} {\bibfnamefont {J.}~\bibnamefont {Kim}}, \bibinfo {author} {\bibfnamefont {K.~P.}\ \bibnamefont {Esler}}, \bibinfo {author} {\bibfnamefont {J.}~\bibnamefont {McMinis}}, \bibinfo {author} {\bibfnamefont {M.~A.}\ \bibnamefont {Morales}}, \bibinfo {author} {\bibfnamefont {B.~K.}\ \bibnamefont {Clark}}, \bibinfo {author} {\bibfnamefont {L.}~\bibnamefont {Shulenburger}}, \ and\ \bibinfo {author} {\bibfnamefont {D.~M.}\ \bibnamefont {Ceperley}},\ }in\ \href@noop {} {\emph {\bibinfo {booktitle} {Journal of Physics: Conference Series}}},\ Vol.\ \bibinfo {volume} {402}\ (\bibinfo {organization} {IOP Publishing},\ \bibinfo {year} {2012})\ p.\ \bibinfo {pages} {012008}\BibitemShut {NoStop}%
\bibitem [{\citenamefont {Cain}\ \emph {et~al.}(2023)\citenamefont {Cain}, \citenamefont {Chattopadhyay}, \citenamefont {Liu}, \citenamefont {Samajdar}, \citenamefont {Pichler},\ and\ \citenamefont {Lukin}}]{cain2023quantum}%
  \BibitemOpen
  \bibfield  {author} {\bibinfo {author} {\bibfnamefont {M.}~\bibnamefont {Cain}}, \bibinfo {author} {\bibfnamefont {S.}~\bibnamefont {Chattopadhyay}}, \bibinfo {author} {\bibfnamefont {J.-G.}\ \bibnamefont {Liu}}, \bibinfo {author} {\bibfnamefont {R.}~\bibnamefont {Samajdar}}, \bibinfo {author} {\bibfnamefont {H.}~\bibnamefont {Pichler}}, \ and\ \bibinfo {author} {\bibfnamefont {M.~D.}\ \bibnamefont {Lukin}},\ }\href@noop {} {\bibfield  {journal} {\bibinfo  {journal} {arXiv preprint}\ } (\bibinfo {year} {2023})}\BibitemShut {NoStop}%
\bibitem [{\citenamefont {Childs}\ \emph {et~al.}(2003)\citenamefont {Childs}, \citenamefont {Cleve}, \citenamefont {Deotto}, \citenamefont {Farhi}, \citenamefont {Gutmann},\ and\ \citenamefont {Spielman}}]{childs2003exponential}%
  \BibitemOpen
  \bibfield  {author} {\bibinfo {author} {\bibfnamefont {A.~M.}\ \bibnamefont {Childs}}, \bibinfo {author} {\bibfnamefont {R.}~\bibnamefont {Cleve}}, \bibinfo {author} {\bibfnamefont {E.}~\bibnamefont {Deotto}}, \bibinfo {author} {\bibfnamefont {E.}~\bibnamefont {Farhi}}, \bibinfo {author} {\bibfnamefont {S.}~\bibnamefont {Gutmann}}, \ and\ \bibinfo {author} {\bibfnamefont {D.~A.}\ \bibnamefont {Spielman}},\ }in\ \href@noop {} {\emph {\bibinfo {booktitle} {Proceedings of the thirty-fifth annual ACM symposium on Theory of computing}}}\ (\bibinfo {year} {2003})\ pp.\ \bibinfo {pages} {59--68}\BibitemShut {NoStop}%
\bibitem [{\citenamefont {Chakraborty}\ \emph {et~al.}(2020)\citenamefont {Chakraborty}, \citenamefont {Luh},\ and\ \citenamefont {Roland}}]{chakraborty2020fast}%
  \BibitemOpen
  \bibfield  {author} {\bibinfo {author} {\bibfnamefont {S.}~\bibnamefont {Chakraborty}}, \bibinfo {author} {\bibfnamefont {K.}~\bibnamefont {Luh}}, \ and\ \bibinfo {author} {\bibfnamefont {J.}~\bibnamefont {Roland}},\ }\href@noop {} {\bibfield  {journal} {\bibinfo  {journal} {Physical review letters}\ }\textbf {\bibinfo {volume} {124}},\ \bibinfo {pages} {050501} (\bibinfo {year} {2020})}\BibitemShut {NoStop}%
\bibitem [{\citenamefont {Richter}(2007{\natexlab{a}})}]{Richter_2007}%
  \BibitemOpen
  \bibfield  {author} {\bibinfo {author} {\bibfnamefont {P.~C.}\ \bibnamefont {Richter}},\ }\href@noop {} {\bibfield  {journal} {\bibinfo  {journal} {New Journal of Physics}\ }\textbf {\bibinfo {volume} {9}},\ \bibinfo {pages} {72} (\bibinfo {year} {2007}{\natexlab{a}})}\BibitemShut {NoStop}%
\bibitem [{\citenamefont {Ding}\ \emph {et~al.}(2009)\citenamefont {Ding}, \citenamefont {Lubetzky},\ and\ \citenamefont {Peres}}]{ding2009mixing}%
  \BibitemOpen
  \bibfield  {author} {\bibinfo {author} {\bibfnamefont {J.}~\bibnamefont {Ding}}, \bibinfo {author} {\bibfnamefont {E.}~\bibnamefont {Lubetzky}}, \ and\ \bibinfo {author} {\bibfnamefont {Y.}~\bibnamefont {Peres}},\ }\href@noop {} {\bibfield  {journal} {\bibinfo  {journal} {Communications in Mathematical Physics}\ }\textbf {\bibinfo {volume} {289}},\ \bibinfo {pages} {725} (\bibinfo {year} {2009})}\BibitemShut {NoStop}%
\bibitem [{\citenamefont {Saloff-Coste}(2004)}]{saloff2004random}%
  \BibitemOpen
  \bibfield  {author} {\bibinfo {author} {\bibfnamefont {L.}~\bibnamefont {Saloff-Coste}},\ }in\ \href@noop {} {\emph {\bibinfo {booktitle} {Probability on discrete structures}}}\ (\bibinfo  {publisher} {Springer},\ \bibinfo {year} {2004})\ pp.\ \bibinfo {pages} {263--346}\BibitemShut {NoStop}%
\bibitem [{\citenamefont {Gerhardt}\ and\ \citenamefont {Watrous}(2003)}]{gerhardt2003continuous}%
  \BibitemOpen
  \bibfield  {author} {\bibinfo {author} {\bibfnamefont {H.}~\bibnamefont {Gerhardt}}\ and\ \bibinfo {author} {\bibfnamefont {J.}~\bibnamefont {Watrous}}\ }(\bibinfo {organization} {Springer},\ \bibinfo {year} {2003})\ pp.\ \bibinfo {pages} {290--301}\BibitemShut {NoStop}%
\bibitem [{\citenamefont {Banerjee}(2022)}]{banerjee2022discrete}%
  \BibitemOpen
  \bibfield  {author} {\bibinfo {author} {\bibfnamefont {A.}~\bibnamefont {Banerjee}},\ }\href@noop {} {\bibfield  {journal} {\bibinfo  {journal} {arXiv preprint}\ } (\bibinfo {year} {2022})}\BibitemShut {NoStop}%
\bibitem [{\citenamefont {Venegas-Andraca}(2012)}]{venegas2012quantum}%
  \BibitemOpen
  \bibfield  {author} {\bibinfo {author} {\bibfnamefont {S.~E.}\ \bibnamefont {Venegas-Andraca}},\ }\href@noop {} {\bibfield  {journal} {\bibinfo  {journal} {Quantum Information Processing}\ }\textbf {\bibinfo {volume} {11}},\ \bibinfo {pages} {1015} (\bibinfo {year} {2012})}\BibitemShut {NoStop}%
\bibitem [{\citenamefont {Kadian}\ \emph {et~al.}(2021)\citenamefont {Kadian}, \citenamefont {Garhwal},\ and\ \citenamefont {Kumar}}]{kadian2021quantum}%
  \BibitemOpen
  \bibfield  {author} {\bibinfo {author} {\bibfnamefont {K.}~\bibnamefont {Kadian}}, \bibinfo {author} {\bibfnamefont {S.}~\bibnamefont {Garhwal}}, \ and\ \bibinfo {author} {\bibfnamefont {A.}~\bibnamefont {Kumar}},\ }\href@noop {} {\bibfield  {journal} {\bibinfo  {journal} {Computer Science Review}\ }\textbf {\bibinfo {volume} {41}},\ \bibinfo {pages} {100419} (\bibinfo {year} {2021})}\BibitemShut {NoStop}%
\bibitem [{\citenamefont {Ambainis}\ \emph {et~al.}(2001)\citenamefont {Ambainis}, \citenamefont {Bach}, \citenamefont {Nayak}, \citenamefont {Vishwanath},\ and\ \citenamefont {Watrous}}]{ambainis2001one}%
  \BibitemOpen
  \bibfield  {author} {\bibinfo {author} {\bibfnamefont {A.}~\bibnamefont {Ambainis}}, \bibinfo {author} {\bibfnamefont {E.}~\bibnamefont {Bach}}, \bibinfo {author} {\bibfnamefont {A.}~\bibnamefont {Nayak}}, \bibinfo {author} {\bibfnamefont {A.}~\bibnamefont {Vishwanath}}, \ and\ \bibinfo {author} {\bibfnamefont {J.}~\bibnamefont {Watrous}},\ }in\ \href@noop {} {\emph {\bibinfo {booktitle} {Proceedings of the thirty-third annual ACM symposium on Theory of computing}}}\ (\bibinfo {year} {2001})\ pp.\ \bibinfo {pages} {37--49}\BibitemShut {NoStop}%
\bibitem [{\citenamefont {Richter}(2007{\natexlab{b}})}]{PR1}%
  \BibitemOpen
  \bibfield  {author} {\bibinfo {author} {\bibfnamefont {P.~C.}\ \bibnamefont {Richter}},\ }\href@noop {} {\bibfield  {journal} {\bibinfo  {journal} {Physical Review A}\ }\textbf {\bibinfo {volume} {76}},\ \bibinfo {pages} {042306} (\bibinfo {year} {2007}{\natexlab{b}})}\BibitemShut {NoStop}%
\bibitem [{\citenamefont {Sin}\ and\ \citenamefont {Sorci}(2022{\natexlab{a}})}]{sin2022continuoustimequantumwalkscayley}%
  \BibitemOpen
  \bibfield  {author} {\bibinfo {author} {\bibfnamefont {P.}~\bibnamefont {Sin}}\ and\ \bibinfo {author} {\bibfnamefont {J.}~\bibnamefont {Sorci}},\ }\href@noop {} {\enquote {\bibinfo {title} {Continuous-time quantum walks on cayley graphs of extraspecial groups},}\ } (\bibinfo {year} {2022}{\natexlab{a}}),\ \Eprint {http://arxiv.org/abs/2011.07566} {arXiv:2011.07566 [math.CO]} \BibitemShut {NoStop}%
\bibitem [{\citenamefont {Wang}\ \emph {et~al.}(2024)\citenamefont {Wang}, \citenamefont {Jiang},\ and\ \citenamefont {Li}}]{wang2024unifyingquantumspatialsearch}%
  \BibitemOpen
  \bibfield  {author} {\bibinfo {author} {\bibfnamefont {Q.}~\bibnamefont {Wang}}, \bibinfo {author} {\bibfnamefont {Y.}~\bibnamefont {Jiang}}, \ and\ \bibinfo {author} {\bibfnamefont {L.}~\bibnamefont {Li}},\ }\href@noop {} {\enquote {\bibinfo {title} {Unifying quantum spatial search, state transfer and uniform sampling on graphs: simple and exact},}\ } (\bibinfo {year} {2024}),\ \Eprint {http://arxiv.org/abs/2407.02530} {arXiv:2407.02530 [quant-ph]} \BibitemShut {NoStop}%
\bibitem [{\citenamefont {Sin}\ and\ \citenamefont {Sorci}(2022{\natexlab{b}})}]{sin}%
  \BibitemOpen
  \bibfield  {author} {\bibinfo {author} {\bibfnamefont {P.}~\bibnamefont {Sin}}\ and\ \bibinfo {author} {\bibfnamefont {J.}~\bibnamefont {Sorci}},\ }\href {\doibase 10.5802/alco.237} {\bibfield  {journal} {\bibinfo  {journal} {Algebraic Combinatorics}\ }\textbf {\bibinfo {volume} {5}},\ \bibinfo {pages} {699} (\bibinfo {year} {2022}{\natexlab{b}})}\BibitemShut {NoStop}%
\bibitem [{\citenamefont {Dai}\ \emph {et~al.}(2018)\citenamefont {Dai}, \citenamefont {Yuan},\ and\ \citenamefont {Li}}]{Discrete_walk_Dn}%
  \BibitemOpen
  \bibfield  {author} {\bibinfo {author} {\bibfnamefont {W.}~\bibnamefont {Dai}}, \bibinfo {author} {\bibfnamefont {J.}~\bibnamefont {Yuan}}, \ and\ \bibinfo {author} {\bibfnamefont {D.}~\bibnamefont {Li}},\ }\href {\doibase 10.48550/ARXIV.1810.00158} {\enquote {\bibinfo {title} {Discrete-time quantum walk on the cayley graph of the dihedral group},}\ } (\bibinfo {year} {2018})\BibitemShut {NoStop}%
\bibitem [{\citenamefont {Cao}\ and\ \citenamefont {Feng}(2021)}]{cao2021perfect}%
  \BibitemOpen
  \bibfield  {author} {\bibinfo {author} {\bibfnamefont {X.}~\bibnamefont {Cao}}\ and\ \bibinfo {author} {\bibfnamefont {K.}~\bibnamefont {Feng}},\ }\href@noop {} {\bibfield  {journal} {\bibinfo  {journal} {Linear and Multilinear Algebra}\ }\textbf {\bibinfo {volume} {69}},\ \bibinfo {pages} {343} (\bibinfo {year} {2021})}\BibitemShut {NoStop}%
\bibitem [{\citenamefont {Liu}\ \emph {et~al.}(2020)\citenamefont {Liu}, \citenamefont {Yuan}, \citenamefont {Dai},\ and\ \citenamefont {Li}}]{three_state_QW}%
  \BibitemOpen
  \bibfield  {author} {\bibinfo {author} {\bibfnamefont {Y.}~\bibnamefont {Liu}}, \bibinfo {author} {\bibfnamefont {J.}~\bibnamefont {Yuan}}, \bibinfo {author} {\bibfnamefont {W.}~\bibnamefont {Dai}}, \ and\ \bibinfo {author} {\bibfnamefont {D.}~\bibnamefont {Li}},\ }\href {\doibase 10.48550/ARXIV.2006.08992} {\enquote {\bibinfo {title} {Three-state quantum walk on the cayley graph of the dihedral group},}\ } (\bibinfo {year} {2020})\BibitemShut {NoStop}%
\bibitem [{\citenamefont {Kendon}(2007)}]{kendon2007decoherence}%
  \BibitemOpen
  \bibfield  {author} {\bibinfo {author} {\bibfnamefont {V.}~\bibnamefont {Kendon}},\ }\href@noop {} {\bibfield  {journal} {\bibinfo  {journal} {Mathematical structures in computer science}\ }\textbf {\bibinfo {volume} {17}},\ \bibinfo {pages} {1169} (\bibinfo {year} {2007})}\BibitemShut {NoStop}%
\bibitem [{\citenamefont {Gallian}(2021)}]{gallian2021contemporary}%
  \BibitemOpen
  \bibfield  {author} {\bibinfo {author} {\bibfnamefont {J.}~\bibnamefont {Gallian}},\ }\href@noop {} {\emph {\bibinfo {title} {Contemporary abstract algebra}}}\ (\bibinfo  {publisher} {Chapman and Hall/CRC},\ \bibinfo {year} {2021})\BibitemShut {NoStop}%
\bibitem [{\citenamefont {Levin}\ and\ \citenamefont {Peres}(2017)}]{levin2017markov}%
  \BibitemOpen
  \bibfield  {author} {\bibinfo {author} {\bibfnamefont {D.~A.}\ \bibnamefont {Levin}}\ and\ \bibinfo {author} {\bibfnamefont {Y.}~\bibnamefont {Peres}},\ }\href@noop {} {\emph {\bibinfo {title} {Markov chains and mixing times}}},\ Vol.\ \bibinfo {volume} {107}\ (\bibinfo  {publisher} {American Mathematical Soc.},\ \bibinfo {year} {2017})\BibitemShut {NoStop}%
\bibitem [{\citenamefont {Gao}\ and\ \citenamefont {Luo}(2010)}]{GAO}%
  \BibitemOpen
  \bibfield  {author} {\bibinfo {author} {\bibfnamefont {X.}~\bibnamefont {Gao}}\ and\ \bibinfo {author} {\bibfnamefont {Y.}~\bibnamefont {Luo}},\ }\href {\doibase https://doi.org/10.1016/j.laa.2009.12.040} {\bibfield  {journal} {\bibinfo  {journal} {Linear Algebra and its Applications}\ }\textbf {\bibinfo {volume} {432}},\ \bibinfo {pages} {2974} (\bibinfo {year} {2010})}\BibitemShut {NoStop}%
\bibitem [{\citenamefont {Richter}(2007{\natexlab{c}})}]{richter2007quantum}%
  \BibitemOpen
  \bibfield  {author} {\bibinfo {author} {\bibfnamefont {P.~C.}\ \bibnamefont {Richter}},\ }\href@noop {} {\emph {\bibinfo {title} {Quantum walks and ground state problems}}}\ (\bibinfo  {publisher} {Rutgers The State University of New Jersey, School of Graduate Studies},\ \bibinfo {year} {2007})\BibitemShut {NoStop}%
\bibitem [{\citenamefont {Aharonov}\ \emph {et~al.}(2001)\citenamefont {Aharonov}, \citenamefont {Ambainis}, \citenamefont {Kempe},\ and\ \citenamefont {Vazirani}}]{aharonov2001quantum}%
  \BibitemOpen
  \bibfield  {author} {\bibinfo {author} {\bibfnamefont {D.}~\bibnamefont {Aharonov}}, \bibinfo {author} {\bibfnamefont {A.}~\bibnamefont {Ambainis}}, \bibinfo {author} {\bibfnamefont {J.}~\bibnamefont {Kempe}}, \ and\ \bibinfo {author} {\bibfnamefont {U.}~\bibnamefont {Vazirani}},\ }in\ \href@noop {} {\emph {\bibinfo {booktitle} {Proceedings of the thirty-third annual ACM symposium on Theory of computing}}}\ (\bibinfo {year} {2001})\ pp.\ \bibinfo {pages} {50--59}\BibitemShut {NoStop}%
\bibitem [{\citenamefont {Childs}\ and\ \citenamefont {Eisenberg}(2003)}]{childs2003quantum}%
  \BibitemOpen
  \bibfield  {author} {\bibinfo {author} {\bibfnamefont {A.~M.}\ \bibnamefont {Childs}}\ and\ \bibinfo {author} {\bibfnamefont {J.~M.}\ \bibnamefont {Eisenberg}},\ }\href@noop {} {\bibfield  {journal} {\bibinfo  {journal} {arXiv preprint}\ } (\bibinfo {year} {2003})}\BibitemShut {NoStop}%
\end{thebibliography}%

\onecolumngrid

\newpage
\appendix

\section{S\MakeLowercase{upplementary material for ``Quantum walks advantage on the dihedral group for uniform sampling problem".}}

\subsection{{\em Q\MakeLowercase{uantum walk algorithm}}}\label{qalgo}

We study a quantum walk on the graph $\Gamma(D_{2n}, S)$, where $S = \{a, a^{-1}, b\}$. The graph is 3-regular and has a vertex set $V = \{1, a, a^2, \dots, a^{n-1}, b, ba, ba^2, \dots, ba^{n-1}\}$. The edge set $E$ consists of pairs $\{g, h\}$ for all $g,h \in D_{2n}$ such that $gh^{-1} \in S$. We define the normalized adjacency matrix $\Bar{A}$ of $\Gamma$ such that $\Bar{A}(i, j) = 1/3$ if vertex $i$ is adjacent to vertex $j$ in $\Gamma$, and $0$ otherwise. The continuous-time quantum walk operator for a given time $t$ is denoted as $U(t)$ and is defined as $\text{e}^{\text{i} \Bar{A} t}$. The quantum walk algorithm starts from an initial state $\ket{x_0} = \ket{p}$, where $p$ is a vertex from the vertex set $V$. The algorithm then performs $T T'$ steps as specified.

\bigskip
\noindent\hrule
\begin{center}
\textbf{Algorithm 1}
\end{center}
\noindent\hrule
\begin{itemize}
    \item Quantum walk algorithm $(p, T',T)$
    \begin{enumerate}
        \item $r = 0$; $\ket{x_0}=\ket{p}$;
        \item While $(T' \geq r)$ 
        \begin{itemize}
            \item Perform the quantum walk starting with $\ket{x_r}$ for time $t$ chosen uniformly at random from $[0,T]$;
            
        \item Let $\ket{\psi_{r+1}}= \text{e}^{\text{i}\Bar{A}t} \ket{x_r}$;
        \item Measure $\ket{\psi_{r+1}}$ in the position basis and obtain the state $\ket{x_{r+1}}$;
        \item $r = r+1$ ;
        \end{itemize}
    \item Output $\ket{x_r}$
    \end{enumerate}
\end{itemize}
\hrule
\bigskip
\noindent 

\subsection{{\em B\MakeLowercase{ounds for the cases in} T\MakeLowercase{heorem~\ref{maintheo}}}}\label{ap2}

\textit{ In this section, we give a detailed analysis of bounds on the sum of the inverse of eigenvalue gaps used in Theorem~\ref{maintheo}.}

\[\textbf{Case 1: $j \in [0,\lfloor{n/4}\rfloor]$ and $ k \in [0,\lfloor{n/4}\rfloor]$ }.\]

In the given range, the following inequalities hold. 

\begin{equation}
    \cos{\Big(\frac{2 \pi j}{n}\Big)} \geq -\frac{2}{\pi}\Big(\frac{2 \pi j}{n}\Big) + 1,
\end{equation}

and 
\begin{equation}
  -\cos{\Big(\frac{2 \pi k}{n}\Big)}\geq -1. 
\end{equation}

This implies 

\begin{align}
 \Big|\cos{\Big(\frac{2 \pi j}{n}\Big)} -  \cos{\Big(\frac{2 \pi k}{n}\Big)} + 1 \Big| & \geq \Big|-\frac{2}{\pi}\Big(\frac{2 \pi j}{n}\Big) + 1 \Big|  \\
 & = \Big| -\frac{4j}{n} + 1 \Big|\\
 & \geq \Big| \frac{n-4j}{n}\Big| .
\end{align}

Hence, the bound on the required sum is

\begin{equation}
 \frac{3}{2}\sum_{j = 0}^{\lfloor\frac{n}{4}\rfloor} \sum_{k = 0}^{\lfloor\frac{n}{4}\rfloor} \frac{1}{|\cos{(2\pi j/n)} - \cos{(2\pi k)/n)}+1 |} \leq  \frac{3}{2}\sum_{j = 0}^{\lfloor\frac{n}{4}\rfloor} \sum_{k = 0}^{\lfloor\frac{n}{4}\rfloor} \frac{n}{|n-4j|}, 
\end{equation}

which is less than 
\begin{equation}
    \begin{aligned}
       \sum_{j = 0}^{\lfloor\frac{n}{4}\rfloor} \frac{1}{|n-4j|} &\leq \frac{1}{n-4\lfloor \frac{n}{4} \rfloor} + \int_{0}^{\lfloor \frac{n}{4} \rfloor} \frac{dj}{n-4j} \\
       &= \frac{1}{4} \Big[\ln{(n)}-\ln{\Big(n-4\Big\lfloor \frac{n}{4} \Big\rfloor\Big)} + \frac{1}{\frac{n}{4}-\lfloor \frac{n}{4} \rfloor} \Big]\\
       &\leq  \ln{(n)}.
    \end{aligned}
\end{equation}
So
\begin{align}
 \frac{3}{2}\sum_{j = 0}^{\lfloor\frac{n}{4}\rfloor} \sum_{k = 0}^{\lfloor\frac{n}{4}\rfloor} \frac{1}{|\cos{(2\pi j/n)} - \cos{(2\pi k)/n)}+1 |} \leq \frac{3}{8} n^2 \log{(n)}. 
\end{align}

\[\textbf{Case 2: $j \in [0,\lfloor{n/4}\rfloor]$ and $ k \in [\lceil{n/4\rceil}, (n-1)/2]$ }.\]

On a similar line, we give the following inequalities to bound the sum in this range.

\begin{equation}
    \cos{\Big(\frac{2 \pi j}{n}\Big)} \geq -\frac{2}{\pi}\Big(\frac{2 \pi j}{n}\Big) + 1,
\end{equation}
and
\begin{equation}
    -\cos{\Big(\frac{2 \pi k}{n}\Big)} \geq \frac{2}{\pi}\Big(\frac{2 \pi k}{n}\Big) - 1,
\end{equation}

\begin{align}
 \Big|\cos{\Big(\frac{2 \pi j}{n}\Big)} -  \cos{\Big(\frac{2 \pi k}{n}\Big)} + 1 \Big| & \geq \Big| \frac{4(k-j)}{n} + 1 \Big|.  
\end{align}

\begin{equation}
 \frac{3}{2}\sum_{j = 0}^{\lfloor\frac{n}{4}\rfloor} \sum_{k = \lceil\frac{n}{4}\rceil}^{(n-1)/2} \frac{1}{|\cos{(2\pi j/n)} - \cos{(2\pi k)/n)}+1 |} \leq  \frac{3}{2}\sum_{j = 0}^{\lfloor\frac{n}{4}\rfloor} \sum_{k = \lceil\frac{n}{4}\rceil}^{(n-1)/2} \frac{n}{|n + 4(k-j)|}. 
\end{equation}

\begin{equation}
    \frac{3}{2}\sum_{j = 0}^{\lfloor\frac{n}{4}\rfloor} \sum_{k = \lceil\frac{n}{4}\rceil}^{(n-1)/2} \frac{n}{|n + 4(k-j)|}  \leq \frac{3n^2}{32} \left(\frac{n}{n+4}\right) \leq \frac{3n^2}{32}.
\end{equation}

\[\textbf{Case 3: $j \in [\lceil{n/4\rceil}, (n-1)/2]$ and $ k \in [0,\lfloor{n/4\rfloor}]$ }.\]

The proof of this case is given numerically in the Supplementary material C.

\[\textbf{Case 4: $j \in [\lceil{n/4\rceil}, (n-1)/2]$ and $ k \in [\lceil{n/4\rceil}, (n-1)/2]$ }.\]

We provide the following argument in this range.

\begin{equation}
    \cos{\Big(\frac{2 \pi j}{n}\Big)} + 1 \geq 0,
\end{equation}
and 
\begin{equation}
    -\cos{\Big(\frac{2 \pi k}{n}\Big)} \geq \frac{2}{\pi}\Big(\frac{2 \pi k}{n}\Big) - 1.
\end{equation}

This gives 

\begin{align}
 \Big|\cos{\Big(\frac{2 \pi j}{n}\Big)} -  \cos{\Big(\frac{2 \pi k}{n}\Big)} + 1 \Big| & \geq \Big| \frac{2}{\pi}\Big(\frac{2 \pi k}{n}\Big) - 1 \Big|.  
\end{align}

The sum is bounded as follows:

\begin{align}
     \frac{3}{2}\sum_{j = \lceil\frac{n}{4}\rceil}^{(n-1)/2} \sum_{k = \lceil\frac{n}{4}\rceil}^{(n-1)/2} \frac{1}{|\cos{(2\pi j/n)} - \cos{(2\pi k)/n)}+1 |} &\leq \frac{3}{2}\sum_{j = \lceil\frac{n}{4}\rceil}^{(n-1)/2} \sum_{k = \lceil\frac{n}{4}\rceil}^{(n-1)/2} \left|\frac{n}{n-4k}\right|\\
     & \leq \frac{3}{32} n^2 \log{(n)}.
\end{align}

\[\textbf{Case 5: $j \in [0, (n-1)/2]$ and $ k \in [0, (n-1)/2]$ $\lambda_j \neq \lambda_k$ }.\]

We try to bound the following sum

\begin{equation}
    \sum_{j \in C_1} \sum_{k \in C_1,\lambda_j \neq \lambda_k} \frac{1}{|\lambda_j - \lambda_k|}.
\end{equation}

It can be written as 

\begin{align}
  \sum_{j \in C_1} \sum_{k \in C_1,\lambda_j \neq \lambda_k} \frac{1}{|\lambda_j - \lambda_k|} & = \sum_{j = 0}^{(n-1)/2} \sum_{k = 0, j \neq k}^{(n-1)/2} \frac{1}{ |\cos{(2\pi j/n)} - \cos{(2\pi k)/n)}|}\\
  & \leq 2 \sum_{j = 0}^{(n-1)/2} \sum_{ k = 0, k > j}^{(n-1)/2} \frac{1}{ |\cos{(2\pi j/n)} - \cos{(2\pi k)/n)}|}\\
  & \leq 2\sum_{y = 1}^{(n-1)} \sum_{ z = 1, k > j}^{(n-1)} \frac{1}{|\sin{(\pi y/n)}\sin{(\pi z/n)}|}.
\end{align}

Note that the map $(j,k) \mapsto (j+k, k-j)$ from 
$\{(j,k): 0 \leq j < k \leq (n-1)/2\}$ to
$\{1,2,\dots, n -1\}\times \{1,2, \dots, n-1\}$ (its inverse is
$(y,z) \mapsto ((y-z)/2, (y+z)/2)$). Now consider the sum

\begin{equation}
 \sum_{y = 1}^{n-1} \frac{1}{\sin{(\pi y/n)}} = 2\sum_{y = 1}^{(n-1)/2} \frac{1}{\sin{(\pi y/n)}}. 
\end{equation}

Note that $\theta \in (0, \pi/2) $ $\implies$ $ \frac{\theta}{2} \leq \sin{\theta}$ $\implies$ $ \frac{2}{\theta} \geq \frac{1}{\sin{\theta}}$ . So for $y \in [1, \frac{n-1}{2}] $ $\implies$ $\frac{\pi y}{n} \in [\frac{\pi}{n}, \pi/2).$ This implies 

\begin{align}
  2\sum_{y = 1}^{(n-1)/2} \frac{1}{\sin{(\pi y/n)}}  & \leq 4\sum_{y = 1}^{(n-1)/2} \frac{n}{\pi y} \\
  & \leq \frac{4n}{\pi} \Big[1 + \log{\Big(\frac{n-1}{2}\Big)} \Big] \\
  & \leq \frac{8n}{\pi} \log{\Big(\frac{n-1}{2}\Big)}.
\end{align}

Hence 

\begin{align}
    \sum_{y = 1}^{(n-1)} \sum_{ z = 1, k > j}^{(n-1)} \frac{1}{|\sin{(\pi y/n)}\sin{(\pi z/n)}|} & \leq
    \Big(\frac{8n}{\pi} \log{(n)}\Big)^2.
\end{align}

On the same line, 

\begin{equation}
    \sum_{j \in C_2} \sum_{k \in C_2,\lambda_j \neq \lambda_k} \frac{1}{|\lambda_j - \lambda_k|} \leq  \Big(\frac{8n}{\pi} \log{(n)}\Big)^2. 
\end{equation}

\subsection{{\em C\MakeLowercase{onjecture~\ref{conj}} }}\label{conjecture}

\textit{In this section, we propose a conjecture to give an upper bound on the sum $Su_{3}$. Subsequently, we provide a numerical argument to support the conjecture. }

\begin{conjecture}
   Consider $n = 4p+1$ type, where $p \in \mathbb{Z}_{+}$, let $\alpha = \arccos{\Big(1-\sin{\Big(\frac{2\pi}{n}(b+\frac{3}{4})\Big)} \Big)}$, and $N(\alpha) = \lfloor \frac{n}{2\pi} \alpha \rfloor$ for $b \in[0, p-1]$ then

   \begin{align}
    \sum_{k = 0}^{\lfloor \frac{n}{4}\rfloor} \sum_{j = \lceil 
\frac{n}{4} \rceil}^{\frac{n-1}{2}} \frac{1}{\Big|\cos{\Big(\frac{2\pi j}{n} \Big)}-\cos{\Big(\frac{2\pi k}{n} \Big)}+1\Big|} \leq f(n) = \sum_{b=0}^{p-1} f_{\alpha}(b) \leq 100 n^2 (\log{(n)})^5,  
   \end{align}

   where 

   \begin{equation}
    f_{\alpha}(b) = \frac{\pi}{2\alpha}\Bigg[\frac{1}{\alpha} + \frac{1}{\alpha-\frac{2\pi}{n}N(\alpha)} + \frac{1}{\frac{2\pi}{n}(N(\alpha)+1)-\alpha}\Bigg] + \frac{n}{4\alpha} \ln{\Bigg[\frac{\frac{\pi^2}{2}}{ (\alpha-\frac{2\pi}{n}N(\alpha)) (\frac{2\pi}{n}(N(\alpha)+1)-\alpha)} \Bigg]}.  
\end{equation}

Note that when $n = 4p+3$ then we have $\alpha = \arccos{\Big(1-\sin{\Big(\frac{2\pi}{n}(b+\frac{1}{4})\Big)} \Big)}$, and $N(\alpha) = \lfloor \frac{n}{2\pi} \alpha \rfloor$ for $b \in[0, p]$.
\end{conjecture}

\textbf{Numerical argument:}

 Due to $\cos{(x)} = \sin{(x- \pi/2)}$, we have the following 
\begin{equation}
    \sum_{k = 0}^{\lfloor \frac{n}{4}\rfloor} \sum_{j = \lceil 
\frac{n}{4} \rceil}^{\frac{n-1}{2}} \frac{1}{|\cos{\Big(\frac{2\pi j}{n} \Big)}-\cos{\Big(\frac{2\pi k}{n} \Big)}+1|} =  \sum_{k = 0}^{\lfloor \frac{n}{4}\rfloor} \sum_{j = \lceil  \frac{n}{4} \rceil}^{\frac{n-1}{2}} \frac{1}{\Big|\sin{\Big(\frac{2\pi j}{n}-\frac{\pi}{2} \Big)}+\cos{\Big(\frac{2\pi k}{n} \Big)}-1\Big|}.
\end{equation}

For $n=4p+1$, when $j = \lceil \frac{n}{4} \rceil = p + 1$  then $\frac{2\pi}{n}(p+1)-\frac{\pi}{2} = \frac{2\pi}{n}((n-1)/4+1)-\frac{\pi}{2} = \frac{2\pi}{n} \frac{3}{4}$. Subsequently, $\sin{\Big(\frac{2\pi j}{n}-\frac{\pi}{2} \Big)} = \sin{\Big(\frac{2\pi}{n} (b + \frac{3}{4}) \Big)}$ where $b \in [0,p-1]$. Similarly, we change the $k$ range to $a \in [0,p]$. The above sum can then be written in terms of $b$

\begin{equation}
         \sum_{k = 0}^{\lfloor \frac{n}{4}\rfloor} \sum_{j = \lceil  \frac{n}{4} \rceil}^{\frac{n-1}{2}} \frac{1}{|\sin{\Big(\frac{2\pi j}{n}-\frac{\pi}{2} \Big)}+\cos{\Big(\frac{2\pi k}{n} \Big)}-1|} =  \sum_{a = 0}^{p} \sum_{b = 0}^{p-1} \frac{1}{|\sin{\Big(\frac{2\pi}{n} (b+\frac{3}{4}) \Big)}+\cos{\Big(\frac{2\pi a}{n} \Big)}-1|}.
\end{equation}

Note that for any $\theta, \theta' \in ( 0, \pi/2)$,

\begin{equation}
    |1-\cos(\theta)-(1-\cos(\theta'))| \geq \frac{1}{\pi}|\theta^2-\theta'^2|.
\end{equation}

For $\alpha = \arccos{\Big(1-\sin{\frac{2\pi}{n}(b+\frac{3}{4})} \Big)}$,  and $N(\alpha) = \lfloor \frac{n}{2\pi} \alpha \rfloor$
the sum is 

\begin{equation}
    \begin{aligned}
       & \sum_{a = 0}^{p} \sum_{b = 0}^{p-1} \frac{1}{\Big|\sin{\Big(\frac{2\pi}{n} (b+\frac{3}{4}) \Big)}+\cos{\Big(\frac{2\pi a}{n} \Big)}-1\Big|}  \leq \pi \sum_{a = 0}^{p} \sum_{b = 0}^{p-1} \frac{1}{|\alpha^2-(\frac{2\pi}{n} a)^2|} \\
        & = \frac{\pi}{2\alpha} \sum_{b=0}^{p-1} \Bigg[\sum_{a = 0}^{N(\alpha)} \Bigg(\frac{1}{\alpha + \frac{2\pi a}{n}} + \frac{1}{\alpha - \frac{2\pi a}{n}} \Bigg) + \sum_{a = N(\alpha+1)}^{p} \Bigg(\frac{1}{\frac{2\pi a}{n}-\alpha} - \frac{1}{\alpha +\frac{2\pi a}{n}} \Bigg) \Bigg] \\
        & \leq \frac{\pi}{2\alpha} \sum_{b=0}^{p-1} \Bigg[\frac{1}{\alpha} + \int_{0}^{N(\alpha)} \frac{\text{d}a}{\alpha + \frac{2\pi}{n}a}  + \frac{1}{\alpha - \frac{2\pi}{n}N(\alpha)} + \int_{0}^{N(\alpha)} \frac{\text{d}a}{\alpha-\frac{2\pi}{n}a} + \frac{1}{\frac{2\pi}{n}(N(\alpha)+1) - \alpha} + \int_{N(\alpha)+1}^{p} \frac{\text{d}a}{\frac{2\pi a}{n}-\alpha}\Bigg]\\
        & = \frac{\pi}{2\alpha} \sum_{b= 0}^{p-1} \Bigg\{ \frac{1}{\alpha} + \frac{1}{\alpha-\frac{2\pi}{n}N(\alpha)} + \frac{1}{\frac{2\pi}{n}(N(\alpha)+1)-\alpha} + \frac{n}{2\pi} \ln\Bigg(\frac{(\alpha+\frac{2\pi}{n})(\frac{2\pi}{n}p-\alpha)}{\big(\alpha-\frac{2\pi}{n}N(\alpha)\big) \big(\frac{2\pi}{n}(N(\alpha)+1)-\alpha \big)} \Bigg) \Bigg\} \\
        & \leq \sum_{b=0}^{p-1} \Bigg\{\frac{\pi}{2\alpha}\Bigg[\frac{1}{\alpha} + \frac{1}{\alpha-\frac{2\pi}{n}N(\alpha)} + \frac{1}{\frac{2\pi}{n}(N(\alpha)+1)-\alpha}\Bigg] + \frac{n}{4\alpha} \ln{\Bigg(\frac{\frac{\pi^2}{2}}{ \big(\alpha-\frac{2\pi}{n}N(\alpha)\big) \big(\frac{2\pi}{n}(N(\alpha)+1)-\alpha \big)} \Bigg)}  \Bigg\}.
    \end{aligned}
\end{equation}

We define the function $f_{\alpha}(b)$, which is given as

\begin{equation}
    f_{\alpha}(b) = \frac{\pi}{2\alpha}\Bigg[\frac{1}{\alpha} + \frac{1}{\alpha-\frac{2\pi}{n}N(\alpha)} + \frac{1}{\frac{2\pi}{n}(N(\alpha)+1)-\alpha}\Bigg] + \frac{n}{4\alpha} \ln{\Bigg(\frac{\frac{\pi^2}{2}}{ \big(\alpha-\frac{2\pi}{n}N(\alpha)\big) \big(\frac{2\pi}{n}(N(\alpha)+1)-\alpha \big)} \Bigg)},  
\end{equation}
and 
\begin{equation}
    f(n) = \sum_{b=0}^{p-1} f_{\alpha}(b).
\end{equation}

Due to the complexity of $f_{\alpha}(b)$ we justify our argument with numerical results.  We plot $f(n)$ (refer Fig.\ref{fig:4p+1}) and show that it is upper bounded by $100 n^2 (\log{(n)})$. 
\begin{figure}[ht!]
    \centering
    \includegraphics[width=0.65\textwidth,keepaspectratio]{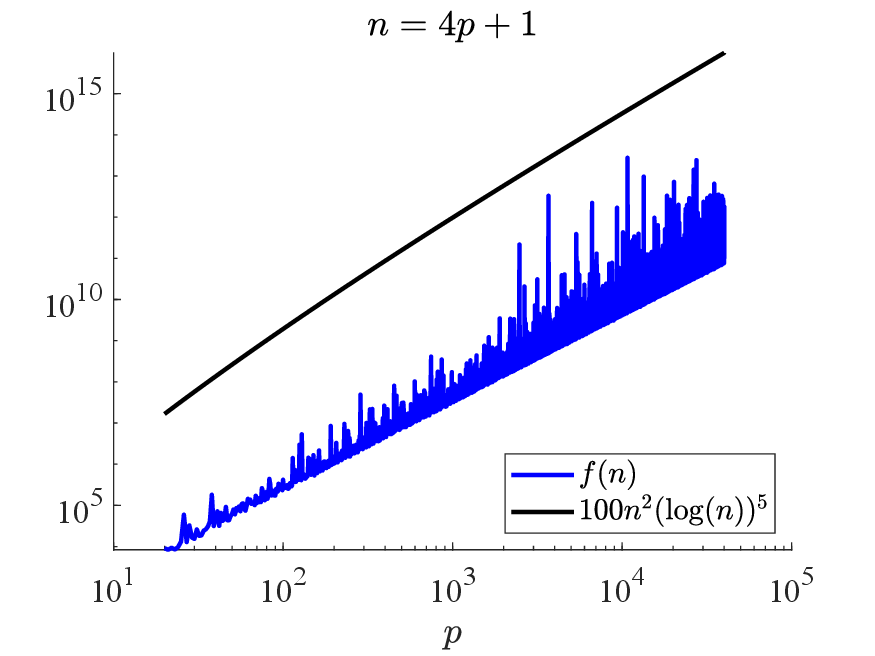} \caption{{{\bf $f(n)$ bound for $n = 4p+1$.} In this plot we show that $f(n)$ is bounded above by $O(n^2 (\log{(n)}))$. Here $x-$ axis scales as $p$, and the $y-$ axis is scaled logarithmically to plot it. We can see that \textcolor{blue}{$y = f(n)$} } is bounded above by  $100 n^2 (\log{(n)})$. }
    \label{fig:4p+1}
\end{figure}
On the same line, we do it for $n = 4p+3$ depicted in FIG.\ref{fig:4p3}. 

\begin{figure}[ht!]
    \centering
    \includegraphics[width=0.65\textwidth,keepaspectratio]{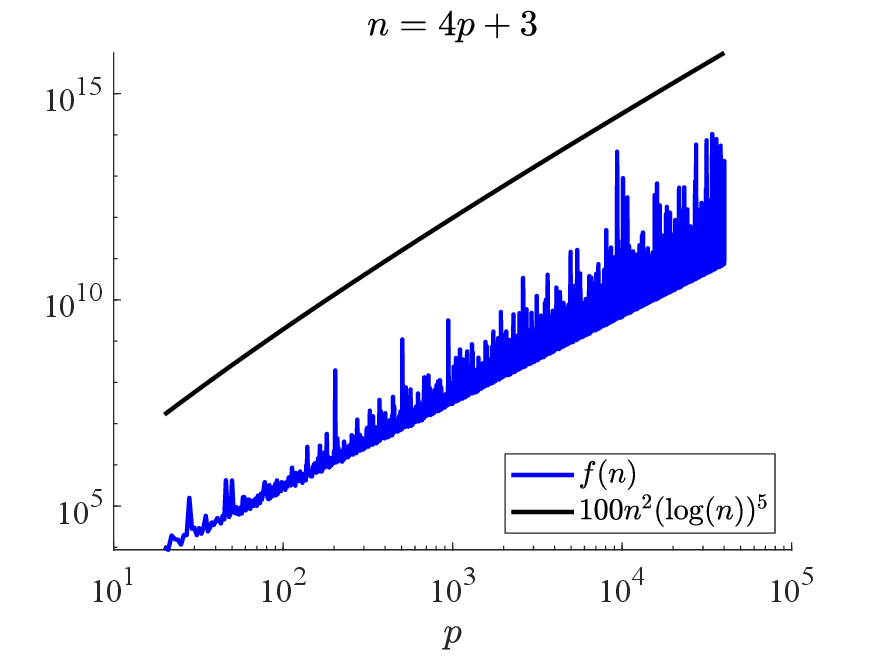} \caption{{{\bf $f(n)$ bound for $n = 4p+3$.} In this plot we show that $f(n)$ is bounded above by $O(n^2 (\log{(n)}))$. Here $x-$ axis scales as $p$, and the $y-$ axis is scaled logarithmically to plot it. We can see that \textcolor{blue}{$y = f(n)$} } is bounded above by  $100 n^2 (\log{(n)})$. }
    \label{fig:4p3}
\end{figure}

\newpage

\subsection{{\em Q\MakeLowercase{uantum walk limiting distribution on $\Gamma(D_{2n}, S)$} }}\label{limitinga}

\textit{In this section, we compute the limiting distribution $\Pi$ of the continuous time quantum walk on the graph $\Gamma(D_{2n}, S)$. }

The probability to go from a vertex $p$ to $q$ on $\Gamma(D_{2n}, S)$ is given Eq.~\eqref{Eq:prob} as follows:
\begin{equation*}
        P_{t}(p,q)=\left|\frac{1}{2n}\bra{q}\sum_{j=0}^{n-1}\text{e}^{\text{i}t(2\cos{(2\pi j/n)}+1)/3}X_j\right.+\sum_{j=n}^{2n-1}\text{e}^{\text{i}t(2\cos{(2\pi (j-n)/n)}-1)/3}Y_j\ket{p}\Bigg|^2. 
\end{equation*}
We change the indexing in the second sum in the above equation from $j$ to $j-n$. This gives,
\begin{equation}\label{Eqn:prob}
    P_{t}(p,q) = \left|\frac{1}{2n}\bra{q}\sum_{j = 0}^{n-1} \text{e}^{\text{i}t(2\cos{(2\pi j/n)+1)/3}} X_j +\sum_{j = 0}^{n-1} \text{e}^{\text{i}t(2\cos{(2\pi j/n)-1)/3}} Y_{j} \right|.
\end{equation}
We now compute the $(p,q)$ entry of time-averaged probability matrix $\Bar{P}_{T}$. It is given by Eq.~\eqref{time-avg}.

First observe that the action of $X_j$ or $Y_j$ on $\ket{p}$ is given as
\begin{equation*}
    (X_j)\ket{p} = \Bar{\omega}^{(p-1)j}\begin{bmatrix}v_j\\v_j\end{bmatrix}= (Y_j)\ket{p}.
\end{equation*}
If we apply $\bra{q}$ on the above state we will get
\begin{equation*}
    \bra{q}(X_j)\ket{p} = \Bar{\omega}^{(p-1)j}\omega^{(q-1)j} = \bra{q}(Y_j)\ket{p} .
\end{equation*}
Now using Eq.~\eqref{time-avg} and \eqref{Eqn:prob} we have,
\begin{equation}
    \begin{aligned}
    \Bar{P}_{T}(p,q) &=\frac{1}{4n^2T}\int_0^T\left|\sum_{j=0}^{n-1}\text{e}^{\text{i}t(2\cos{(2\pi j/n)}+1)/3}w^{(q-p)j}+\sum_{j=0}^{n-1}\text{e}^{\text{i}t(2\cos{(2\pi j/n)}-1)/3}w^{(q-p)j}\right|^2 \text{d}t \\
    &=\frac{1}{4n^{2}T}\int_{0}^{T}\Big|\left(\text{e}^{\text{i}t/3}+\text{e}^{-\text{i}t/3}\right)\sum_{j=0}^{n-1}\text{e}^{\text{i}t(2\cos{(2\pi j/n)})/3}w^{(q-p)j}\Big|^{2} \text{d}t \\
    &=\frac{1}{4n^{2}T}\int_{0}^{T}\left|2\cos{(t/3)}\sum_{j=0}^{n-1}\text{e}^{\text{i}t(2\cos{(2\pi j/n)})/3}w^{(q-p)j}\right|^{2}dt\\
    &=\frac{1}{n^{2}T}\int_{0}^{T} \left(\frac{1+\cos{(2t/3)}}{2}\Big|\sum_{j=0}^{n-1}\text{e}^{\text{i}t(2\cos{(2\pi j/n)})/3}w^{(q-p)j}\Big|^{2}\right)\text{d}t\\
     &=\frac{1}{2n^{2}T}\left(\int_{0}^{T}\Big|\sum_{j=0}^{n-1}\text{e}^{\text{i}t(2\cos{(2\pi j/n)})/3}w^{(q-p)j}\Big|^{2}\text{d}t+\int_{0}^{T}\cos{(2t/3)} \Big|\sum_{j=0}^{n-1}\text{e}^{\text{i}t(2\cos{(2\pi j/n)})/3}w^{(q-p)j}\Big|^{2}\text{d}t\right). \\
     \end{aligned}
    \end{equation}
    
   We have bound on the modulus of the second term because
    
    \begin{equation}
        \frac{1}{2n^{2}T}\left|\int_{0}^{T}\cos{(2t/3)} \Big|\sum_{j=0}^{n-1}\text{e}^{\text{i}t(2\cos{(2\pi j/n)})/3}w^{(q-p)j}\Big|^{2}dt\right|\leq\frac{3}{4n^{2}T} |\sin(2T/3)|(n)^2\leq\frac{3}{4n^{2}T}(n)^2.
    \end{equation}
    
    Since for a complex number $z$, $|z|\rightarrow 0 \text{ iff } z \rightarrow 0$, the second term itself goes to zero as $T$ goes to infinity giving us, 
    \begin{equation}
    \lim_{T \rightarrow \infty}\Bar{P}_{T}(p,q)=\lim_{T \rightarrow \infty} \frac{1}{2n^{2}T}\int_{0}^{T}\Big|\sum_{j=0}^{n-1}\text{e}^{\text{i}t(2\cos{(2\pi j/n)})/3}w^{(q-p)j}\Big|^{2}\text{d}t.
    \end{equation}
   
    For a complex number $z$ we have $|z|^2=z\Bar{z}.$ We use it for $\Big|\sum_{j=0}^{n-1}\text{e}^{\text{i}t(2\cos{(2\pi j/n)})/3}w^{(q-p)j}\Big|^{2}$ and get the following,
    
   \begin{equation}\label{eq:lim prob}
    \begin{aligned} 
    \lim_{T \rightarrow \infty}\Bar{P}_{T}(p,q)& =\lim_{T \rightarrow \infty} \frac{1}{2n^{2}T}\int_{0}^{T}\left(\sum_{j,k=0}^{n-1}\text{e}^{\text{i}t(2\cos{(2\pi j/n)}-2\cos{(2\pi k/n)})/3}w^{(q-p)(j-k)}\right)\text{d}t\\
    &= \lim_{T \rightarrow \infty}\frac{1}{2n^{2}T}\int_{0}^{T}\left(n + \sum_{j,k=0, j \neq k}^{n-1}\text{e}^{\text{i}t(2\cos\left(2\pi j/n\right)-2\cos\left(2\pi k/n\right))/3}w^{(q-p)(j-k)}\right)\text{d}t\\
    &=\frac{1}{2n}+\lim_{T \rightarrow \infty}\frac{1}{2n^{2}T}\int_{0}^{T}\left(\sum_{j,k=0, j \neq k}^{n-1}\text{e}^{\text{i}t(2\cos\left(2\pi j/n\right)-2\cos\left(2\pi k/n\right))/3}w^{(q-p)(j-k)}\right)\text{d}t.
    \end{aligned}
    \end{equation}

    We now make the following cases: 

\[\textbf{Case 1: $p = q \text{ or } q-p = n$ }.\]

    In this case since $w^{(q-p)}=1$, Eq.~\eqref{eq:lim prob} becomes,
    \begin{equation}
    \lim_{T \rightarrow \infty}\Bar{P}_{T}(p,q)=\frac{1}{2n}+\lim_{T \rightarrow \infty}\frac{1}{2n^2T} \int_0^T  \left(\sum_{j,k=0,j\neq k}^{n-1}\text{e}^{\text{i}t(-4\sin{(2\pi (j-k)/n)}\sin{(2\pi (j+k)/n)})/3}\right)\text{d}t.
   \end{equation}

   Now for each $j \geq 1$, there exist $k$ such that $j+k = n$. We separate these terms and get,

    \begin{equation}
    \lim_{T \rightarrow \infty}\Bar{P}_{T}(p,q)=\frac{1}{2n}+  \frac{n-1}{2n^2} + \lim_{T \rightarrow \infty}\frac{1}{2n^2T} \int_0^T  \left(\sum_{j,k=0,j\neq k, j+k \neq n }^{n-1}\text{e}^{\text{i}t(-4\sin{(2\pi (j-k)/n)}\sin{(2\pi (j+k)/n)})/3}\right)\text{d}t.
    \end{equation}

    Finally, since the third term vanishes after evaluating the integral we have,
    \begin{equation}
    \lim_{T \rightarrow \infty}  \Bar{P}_{T}(p,q)  = \frac{1}{2n} +\frac{n-1}{2n^2}.
    \end{equation}

    \[\textbf{Case 2: $p \neq q$ and $q-p\neq n$ }.\]

    Observe that, in the second term of Eq. \eqref{eq:lim prob} for each $j$, $0\leq j \leq n-1$ there exist $k$, $0\leq k \leq n-1$ such that $j+k = n$ and there are $n-1$ such $j$. Separating these terms, we get the geometric sum of $-1$.

    \begin{equation*}
      \begin{aligned}\sum_{j=1}^{n-1}w^{2(q-p)j}&=w^{2(q-p)}\frac{1-w^{2(q-p)(n-1)}}{1-w^{2(q-p)}}=\frac{w^{2(q-p)}-1}{1-w^{2(q-p)}}=-1.
      \end{aligned}
      \end{equation*}
    Using the above value and a formula for the difference of cosines Eq.~\eqref{eq:lim prob} becomes:
    \begin{equation}
     \lim_{T \rightarrow \infty}\Bar{P}_{T}(p,q) =\frac{1}{2n}-\frac{1}{2n^{2}}+ \lim_{T \rightarrow \infty}\frac{1}{2n^{2}T}\int_{0}^{T}\left(\sum_{j,k=0,j\neq k,j+k\neq n}^{n-1}\text{e}^{\text{i}t(-4\sin{(\pi(j+k)/n)}\sin{(\pi(j-k)/n)})/3}w^{(q-p)(j-k)}\right)\text{d}t.
    \end{equation}
    As before, the third term in the above equation vanishes as $T \rightarrow \infty$. Thus,

   \begin{equation}
  \lim_{T \rightarrow \infty}  \Bar{P}_{T}(p,q)= \frac{1}{2n} -\frac{1}{2n^2}.
    \end{equation}

   The two cases together give:
  \begin{equation}
    P_{T\to\infty}(p,q)=\Pi(p,q)=
    \begin{cases}
    \frac{1}{2n}+\frac{n-1}{2n^2} &p=q\quad \text{or} \quad q-p=n,\\
    \frac{1}{2n}-\frac{1}{2n^2} &p\neq q. \quad \text{and} \quad q-p\neq n.
    \end{cases}
  \end{equation}

\end{document}